\documentclass[11pt,a4paper]{amsart}
\usepackage[pdftex]{pict2e}
\usepackage{amsthm}
\usepackage{graphicx}
\usepackage{epstopdf}
\usepackage{indentfirst}
\usepackage[numbers]{natbib}
\usepackage{amstext}
\usepackage{enumitem}

\usepackage{stackrel}
\usepackage{mathpazo}
\usepackage{graphicx}
\usepackage[left=2.2cm,right=2.2cm,top=2cm,bottom=2cm]{geometry}
\usepackage{amsfonts}
\usepackage{amssymb}
\usepackage{amsmath}
\usepackage{latexsym}
\usepackage{mathrsfs}
\usepackage{color}
\usepackage{dsfont}
\usepackage{url,stackrel,enumitem}	
\usepackage{hyperref}
\usepackage{amsfonts}
\usepackage{stmaryrd}
\usepackage{euscript}
\usepackage{amscd}
\usepackage{bm}
\usepackage{subcaption}
\usepackage{multirow}
\usepackage[textwidth=3.2cm]{todonotes}

\newtheorem{lemma}{Lemma}[section]
\newtheorem{proposition}{Proposition}[section]
\newtheorem{corollary}{Corollary}[section]
\newtheorem{remark}{Remark}[section]
\newtheorem{example}{Example}[section]

\newtheorem{definition}{Definition}[section]

\newcommand{\sgn}{\text{\rm sgn}}

\newcommand\I{\mathds{1}}
\newcommand{\lbn}{-1}

\newcommand{\PutO}{P_0}
\newcommand{\Putt}{P_t}
\newcommand{\PutT}{P_T}

\newcommand{\CallO}{C_0}
\newcommand{\Callt}{C_t}
\newcommand{\CallT}{C_T}

\newcommand{\rhodf}{\rho_{12}}
\newcommand{\rhod}{\rho_{13}}
\newcommand{\rhof}{\rho_{23}}

\newcommand{\alone}{\alpha_1}
\newcommand{\altwo}{\alpha_2}
\newcommand{\althree}{\alpha_3}

\newcommand{\Phin}{\Phi}
\newcommand{\pp}{\psi}

\newcommand{\ppfe}{\psi^f}
\newcommand{\pppr}{\psi^p}
\newcommand{\ppf}{\psi_f}
\newcommand{\ppa}{\psi_e}
\newcommand{\ppd}{\psi_d}

\newcommand{\CCF}{\mbox{\rm CF$\,$}}

\newcommand{\wh}{\widehat}
\newcommand{\wt}{\widetilde}
\newcommand{\whc}{\widehat{c}}
\newcommand{\whH}{\widehat{H}}
\newcommand{\whh}{\widehat{h}}
\newcommand{\whY}{\widehat{Y}}

\newcommand{\wtY}{\widetilde{Y}}
\newcommand{\wtK}{\widetilde{K}}

\newcommand{\whd}{\widehat{d}}
\newcommand{\wtd}{\widetilde{d}}
\newcommand{\whC}{\widehat{C}}
\newcommand{\wtC}{\widetilde{C}}

\newcommand{\whP}{\widehat{P}}
\newcommand{\wtP}{\widetilde{P}}

\newcommand{\EPSd}{\mbox{\rm EPS\,}^d_T}
\newcommand{\EPSf}{\mbox{\rm EPS\,}^f_T}
\newcommand{\EPSfe}{\mbox{\rm EPS\,}^{f,e}_T}
\newcommand{\EPSfq}{\mbox{\rm EPS\,}^{f,q}_T}

\newcommand{\EPSwe}{\mbox{\rm EPS\,}^{w,e}_T}

\newcommand{\whHd}{\whH^d_T}

\newcommand{\whHfe}{\whH^{f,e}_T}
\newcommand{\whHfq}{\whH^{f,q}_T}

\newcommand{\whHwe}{\whH^{w,e}_T}

\newcommand{\whHdt}{\whH^d_t}

\newcommand{\whHfet}{\whH^{f,e}_t}
\newcommand{\whHfqt}{\whH^{f,q}_t}

\newcommand{\whHwet}{\whH^{w,e}_t}

\newcommand{\EPSdo}{\mbox{\rm EPS\,}^d_0}
\newcommand{\EPSfo}{\mbox{\rm EPS\,}^f_0}

\newcommand{\EPSfeo}{\mbox{\rm EPS\,}^{f,e}_0}
\newcommand{\EPSfqo}{\mbox{\rm EPS\,}^{f,q}_0}

\newcommand{\X}{Q}

\newcommand{\Np}{N}
\newcommand{\Nd}{N^d}
\newcommand{\Nf}{N^f}
\newcommand{\Nfd}{N^{f,d}}

\newcommand{\Sd}{S^d}
\newcommand{\Sf}{S^f}
\newcommand{\Sfd}{S^{f,e}}

\newcommand{\wtSd}{\wt{S}^{d}}

\newcommand{\wtSfd}{\wt{S}^{f,e}}

\newcommand{\whSd}{\wh{S}^{d}}

\newcommand{\whSfd}{\wh{S}^{f,e}}

\newcommand{\Swe}{S^{w,e}}
\newcommand{\Rwe}{R^{w,e}}

\newcommand{\wtSwe}{\wt{S}^{w,e}}

\newcommand{\Rd}{R^d}
\newcommand{\Rf}{R^f}
\newcommand{\RQ}{R^q}
\newcommand{\Rfd}{R^{f,e}}

\newcommand{\Hwes}{H^{w,e,s}}
\newcommand{\Hwesf}{H^{w,e,s,f}}
\newcommand{\Hwesb}{H^{w,e,s,b}}

\newcommand{\HwesbT}{H_T^{w,e,s,b}}
\newcommand{\Hwea}{H^{w,e,a}}

\newcommand{\Vwe}{X^{w,e}}

\newcommand{\E}{{\mathbb E}}

\newcommand{\Q}{{\mathbb Q}}
\newcommand{\Qd}{{\Q}^d}
\newcommand{\wtQd}{\wt{\Q}^d}
\newcommand{\whQd}{\wh{\Q}^d}

\newcommand{\EQ}{{\E}_{{\Q}}}
\newcommand{\EQd}{{\E}_{{\Q}^d}}

\newcommand{\bbP}{{\mathbb P}}
\newcommand{\bbR}{{\mathbb R}}

\newcommand{\cF}{\mathcal F}

\title{Pricing and Hedging of Cross-Currency \\ Equity Protection Swaps}
\author[M. Rutkowski, H. Xu]{Marek Rutkowski$^{a,b}$, Huansang Xu$^{c}$}

\begin{document}

\maketitle

\begin{center}
\normalsize{\today} \\ \vspace{0.5cm}
\small\textit{$^{a}$ School of Mathematics and Statistics, University of Sydney, \\ Sydney, NSW 2006, Australia}\\
\small\textit{$^{b}$ Faculty of Mathematics and Information Science, Warsaw University of Technology, \\ 00-661 Warszawa, Poland}\\
\small\textit{$^{c}$ National University of Singapore, Department of Mathematics,\\ 21 Lower Kent Ridge Road, 119077.}
\end{center}

%%%%%%%%%%%%%%%%%%%%%%%%%%%%%%%%%%%%%%%%%%%%%%%%%%%%%%%%%%%%%%%%%%%%%%%%%%%%%%%%%%%%%%%%%%%%%%%%%%%%

\begin{abstract}~

We study arbitrage-free pricing and static hedging strategies for an innovative insurance product called the equity protection swap (EPS). Notably, we focus on the application of EPSs involving cross-currency reference portfolios, reflecting the realities of investor asset diversification across different economies. The research examines key considerations regarding exchange rate fluctuations, pricing and hedging frameworks, in order to satisfy dynamic requirements from EPS buyers. We differentiate between two hedging paradigms: one where domestic and foreign equities are treated separately using two EPS products and another that integrates total returns across currencies. Through detailed analysis, we propose various hedging strategies with consideration of different types of returns — nominal, effective, and quanto - for EPS products in both separate and aggregated contexts. The aggregated hedging portfolios contain basket options with cross-currency underlying asset, which only exists in the OTC market, thus we further consider a superhedging strategy. A numerical study assesses hedging costs and performance metrics associated with these hedging strategies, illuminating practical implications for EPS providers and investors engaged in international markets. We further employ the Monte Carlo method for the basket option pricing, together with two other approximation methods - geometric averaging and moment matching. This work contributes to enhancing fair pricing mechanisms and risk management strategies in the evolving landscape of cross-currency financial derivatives.
\\[2mm]

\end{abstract}
% \begin{keywords}
% Equity protection swap; Cross-currency portfolio; Effective return; Static hedging; Basket option.
% \end{keywords}

% \begin{classcode}
% G22, G12, G13, C63
% \end{classcode}

% \newpage

% \newpage
% \tableofcontents
% \newpage

%%%%%%%%%%%%%%%%%%%%%%%%%%%%%%%%%%%%%%%%%%%%%%%%%%%%%%%%%%%%%%%%%%%%%%%
\section{Introduction} \label{sec1}
%%%%%%%%%%%%%%%%%%%%%%%%%%%%%%%%%%%%%%%%%%%%%%%%%%%%%%%%%%%%%%%%%%%%%%%

Xu et al. \cite{XLR2024} have introduced a new type of insurance for superannuation - {\it equity protection swap} (EPS). It is a financial derivative, which is reminiscent of a total return swap but has also some features of the annuity insurance product RILA. Additionally, the structure of an EPS makes it a reliable safeguard when the value of a reference portfolio drops, and it can also be tailored to accommodate the specific requirements of super fund members. Compared with variable annuity riders, an EPS is less complex and is not necessarily associated with a variable annuity. Apart from that, there is a main difference between RILA and EPS products that EPS is an insurance product that can be added to the variable annuity rather than the annuity itself in order to manage the risks.

An important direction in the study of EPSs is the case of a cross-currency reference portfolio, which will be discussed in this paper. In reality, it is fairly common for investors to hold assets across different economies in order to do risk diversification. Then the underlying reference portfolio for an EPS will contain assets from several economies. Therefore, to generalise the application scope of the concept of an EPS, the pricing and hedging strategies should take into account the exchange rate fluctuations, as well as the correlation between prices of domestic and foreign assets. Various existing approaches to the arbitrage-free pricing and hedging of cross-currency reference portfolios include either considering two independent local returns on domestic and foreign assets or focusing on the total return on the portfolio's value expressed in the domestic currency. This area of research is essential for providers and investors operating in international markets and it provides a valuable insight into fair pricing and hedging strategies for an EPS in a global context.

To mitigate the foreign exchange risk, a variant of the Black-Scholes model was first proposed for pricing and hedging of European currency options by Garman and Kohlhagen \cite{GK1983}. It should be acknowledged that since the Garman-Kohlhagen model is based on the geometric Brownian motion, it is unable to correctly specify the dynamics of the exchange rate.

In order to hedge the risk of investments in domestic and foreign reference portfolios, we need to consider the preferences of EPS buyers and examine their respective strategies. If the purchaser regards the underlying portfolio as two separate independent portfolios, they will focus on the actual returns of domestic and foreign equities independently. In this case, we can model two hedging portfolios to hedge domestic and foreign equity risks separately. We need to mention that the hedging portfolio for domestic equity has been discussed before, thus we only need to consider foreign equity risks in this strategy. In another situation, EPS buyers only focus on the total return on the whole reference portfolio rather than separate returns of domestic and foreign equities. The insurance company should hedge total return risks together with currency risks and build only one hedging portfolio, it is more difficult to structure than in the first case.

We aim to analyse the pricing and hedging of an EPS with a cross-currency reference portfolio. For conciseness of notation, we assume that the reference portfolio contains only two economies - a domestic and a foreign one. By convention, the domestic currency is chosen to be the Australian dollar (AUD) and the foreign one is the U.S. dollar (USD) but, obviously, all definitions and results from this paper can be applied to any two (or more than two) currencies.

We start this paper by defining a setup for the domestic and foreign markets in Section \ref{sec2}. We discuss not only the nominal return but also the effective return of the foreign equity, then provide the cross-currency and quanto (i.e., for a guaranteed exchange rate) rate of return for the cross-currency portfolio. The cross-currency reference portfolio of EPS is defined and we discuss the hedging strategies for such EPS products. Unlike for a standard EPS with a single currency reference portfolio studied by Xu et al. \cite{XLR2024}, here we have several different hedging strategies since we consider alternative conventions regarding realised returns and hence various structures of the protection and fee legs.

In Section \ref{sec3} , we study separate hedging strategies, which correspond to the EPS convention where the realised returns of domestic and foreign equities are mitigated independently and, as a consequence, two sets of hedging portfolios with various kinds of domestic and foreign options are used. In fact, in that case we can adapt the static hedging strategies for a standard EPS with a single currency reference portfolio to fully hedge the domestic equity, by adding appropriate weights to the notional principal. However, for the foreign equity, we need to consider not only its realised return but also the returns on the exchange rate. Firstly, the EPS provider can hedge only the nominal return of the foreign equity and exchange this foreign nominal return to domestic return by multiplying the exchange rate at maturity. In this case, European options written on foreign equities can be used directly to build the hedging portfolio and the volatility of the exchange rate will not be considered.

Secondly, one can also hedge the foreign equity's effective return, which is the return on the foreign equity struck in domestic currency. The EPS provider needs to track the value of foreign equity expressed in the domestic currency in the holding period of the EPS product and, in this case, he needs to use European options with foreign equity struck in the domestic currency as the underlying asset. We also consider in Section \ref{sec3.6} the issue of hedging the foreign equity's quanto rate, which predetermines a fixed exchange rate corresponding to the EPS holder's requirements.

In Section \ref{sec4}, we consider the hedging strategy for an EPS referencing the total return that combines the returns of both domestic and foreign equities. We discuss two different hedging strategies here, the first one is a static hedging strategy that uses European cross-currency basket options with the whole portfolio struck in domestic currency as the underlying asset; here we can extend the algorithm of the hedging strategy for standard EPS products with single currency reference portfolio to do the hedging. However, we need to mention that this fully hedging strategy may be unattainable in that European cross-currency basket options on bespoke portfolios are unlikely to be offered in practice, though they can be negotiated on the OTC market. In order to build a more practical hedging portfolio, we further introduce a joint superhedging strategy. By using the subadditive property, we build a joint superhedging strategy by using European options with single currency underlying assets. However, some parts of this joint superhedging strategy still need to be structured in the OTC market.

In Section \ref{sec5}, we present a numerical study for EPS products with the cross-currency reference portfolio.
We select parameters for domestic and foreign markets and examine trends of stock prices and exchange rates with correlated volatility. Our numerical study illustrated the hedging costs for separate hedging strategies: domestic return, effective foreign return, and quanto foreign return, as well as for aggregated hedging strategies: cross-currency strategies, quanto hedging strategies, and superhedging strategies. For the aggregated hedging strategies, we need to use basket options to build the hedging portfolio. However, we cannot find an explicit formula for the basket option and we will use some approximations instead. We will consider the premium obtained using the Monte Carlo method as an exact cost, the closed-form solutions for basket options with geometric averaging approximation and under three-moment matching approximation method are also included as a comparison.

Each hedging strategy satisfies different requirements from EPS buyers, which influence the structure and clauses of the EPS product at the beginning, thus all of them are important and cannot be substituted.
We note that the aggregated superhedging strategy requires higher hedging costs compared with other aggregated hedging strategies, which is consistent with the definition of superhedging.
We conclude this work by summarising our findings and emphasising the practical importance of this study.

%The main difference between notional value and nominal value is that notional value refers to the total value of a position %in financial contracts, while nominal value refers to the face value of an asset or security.

% \newpage

%%%%%%%%%%%%%%%%%%%%%%%%%%%%%%%%%%%%%%%%%%%%%%%%%%%%%%%%%%%%%%%%%%%%%%%%
\section{Generic Equity Protection Swaps} \label{sec2}
%%%%%%%%%%%%%%%%%%%%%%%%%%%%%%%%%%%%%%%%%%%%%%%%%%%%%%%%%%%%%%%%%%%%%%%%

We first define a generic equity protection swap (EPS), which is a financial contract with a single terminal payoff at maturity $T$ determined by the notional principal $N$ and performance of an underlying reference portfolio over the lifetime $[0,T]$ of an EPS. Let us denote by $R_T$ the simple rate of return on the strictly positive {\it wealth process} $X$ of a reference portfolio so that $R_T := (X_T-X_0)/X_0$ satisfies $R_T>-1$.
In the cross-currency framework the portfolio's wealth process $X$, which is given in the domestic currency,
can be formally specified in several alternative ways, depending on the financial interpretation of a reference portfolio and investor's preferences and thus also the random variable $R_T$ can be specified in various ways (see Section \ref{sec3.1}).

Before we state the definition of a generic EPS, we need to introduce some notation. We consider sequences $\lbn = l_{n+1}< l_{n}<\cdots <l_{2}< l_{1}<l_0=0$ and $0=g_{0}< g_{1}< \cdots <g_{m}< g_{m+1} = \infty$ of real numbers representing a predetermined set of thresholds for losses and gains, respectively. For each $k$, the  {\it protection rate} over the loss interval $[l_{k+1},l_k)$ is denoted as $p_{k+1}$ whereas the {\it fee rate} over the gain interval $(g_{k},g_{k+1}]$ is denoted by $f_{k+1}$. It is natural to assume that protection rates belong to the interval $[0,1]$, but it suffices to assume that the fee rates are arbitrary nonnegative constants.

According to the definition of a generic EPS introduced in Xu et al. \cite{XLR2024}, the cash flow at time $T$ from the investor to the provider of an EPS is given by $N \pp (R_T)$ where the function $\pp :(-1,\infty ) \to \bbR$ is used to encode the structure of an EPS and $N$ is the {\it notional principal} of an EPS. We postulate that the payoff profile $\pp$ is a piecewise linear, non-decreasing, continuous function such that $\pp (0)=0$. The following definition was introduced in Xu et al. \cite{XLR2024}.

\begin{definition} \label{def2.1} \rm
A generic EPS is a financial contract with the payoff at maturity date $T$ to its provider given by the \textit{adjusted return} $\pp (R_T)=\pppr(R_T)+\ppfe(R_T)$ where the non-positive {\it protection leg} $\pppr(R_T)$ represents the provider's loss and equals, per one unit of the notional principal $\Np$,
\begin{equation*}
\pppr(R_T):=\pp(R_T)\I_{\{R_T<0\}}=\sum_{i=0}^n \Big(\sum_{k=0}^{i-1}p_{k+1}\Delta l_k+p_{i+1}(R_T-l_i)\Big)\I_{\{R_T \in (l_{i+1},l_i)\}}
\end{equation*}
and the non-negative {\it fee leg} $\ppfe(R_T)$ represents the provider's gain and satisfies
\begin{equation*}
\ppfe(R_T):=\pp(R_T)\I_{\{R_T>0\}}=\sum_{j=0}^m \Big(\sum_{k=0}^{j-1}f_{k+1}\Delta g_k+f_{j+1}(R_T-g_j)\Big)\I_{\{R \in (g_j,g_{j+1})\}}.
\end{equation*}
Equivalently, the EPS provider's payoff is specified by the {\it adjusted return} $\pp (R_T)$ where
the derivative of the adjusted return function $\pp:(-1,\infty ) \to\mathbb{R}$ exists almost everywhere
and is given by the equality
\begin{equation*}
\pp'(R)=\sum_{i=0}^{n} p_{i+1}\I_{(l_{i+1},l_i)}(R)+\sum_{j=0}^{m} f_{j+1}\I_{(g_j,g_{j+1})}(R).
\end{equation*}
\end{definition}

Throughout this work, our main goal is to study arbitrage-free market-based pricing of various classes of
cross-currency EPSs and analyze static hedging strategies based on European options. In particular, we adopt the following natural definition of {\it fairness} of a generic EPS, which explains why the term `swap' is suitable for a generic EPS with fair initial value. Recall that the term {\it swap} is commonly used in reference to a financial contract that is worthless at its inception.

\begin{definition} \label{def2.1a} \rm
 We say that an EPS initiated at time 0 with the payoff structure $\pp$ is {\it fair} if the terminal payoff $\pp (R_T)$ with the settlement date $T$ has null arbitrage-free price at time 0 or, equivalently, if it can be replicated by a static portfolio of European options such that the portfolio's initial cost
is null.
\end{definition}

%%%%%%%%%%%%%%%%%%%%%%%%%%%%%%%%%%%%%%%%%%%%%%%%%%%%%%%%%%%%%%%%%%%%%%%%%
\subsection{Buffer and Floor EPS}   \label{sec2.1}
%%%%%%%%%%%%%%%%%%%%%%%%%%%%%%%%%%%%%%%%%%%%%%%%%%%%%%%%%%%%%%%%%%%%%%%%%

Two EPS products introduced by Xu et al. \cite{XLR2024} will be used throughout as examples of a generic EPS. Needless to say, much more complex products can be offered to cater for individual risk preferences of investors and all of them can be dealt with using general pricing and hedging results from this work.  For simplicity, we only deal here
with EPS products of European style with a payoff at maturity date $T$, although one may allow for cancellation of an EPS by its holder at its current market price provided that the underlying European options are liquidly traded.

We first set $m=n=1$ and define the \textit{buffer EPS} where the concept of a buffer applies to both the protection and fee leg. Let $\text{Call}_{\,T}(X,K,T)=(X_T-K)^+$ and $\text{Put}_{\,T}(X,K,T)=(K-X_T)^+$ denote the payoffs of a European call and put options written on the process $X$,  with the strike $K$ and maturity date $T$.

\begin{definition} \label{def2.2}
{\rm The \textit{buffer EPS} is obtained by taking $p_1=f_1=0$ (hence $\pp_B(R)=0$ for all $R \in [l_1,g_1]$ where $l_1<0$ and $g_1>0$) and some values for $p_2\in (0,1]$ and $f_2>0$. Hence the cash flow at maturity $T$ of a buffer EPS to its provider equals
\begin{align*}
\pp_B (R_T)& :=p_2(R_T-l_1) \I_{\{R_T \in (-1,l_1)\}}+f_2(R_T-g_1)\I_{\{R_T\in (g_1,\infty)\}}\\
&=-p_2(l_1-R_T)^+ +f_2(R_T-g_1)^+ \\
&=\frac{p_2}{X_0}\,\PutT (X,K^l_1(X_0),T)-\frac{f_2}{X_0}\,\CallT (X,K^g_1(X_0),T)
\end{align*}
where $K^l_1(X_0):=(1+l_1)X_0$ and $K^g_1(X_0):=(1+g_1)X_0$.}
\end{definition}

%\begin{figure}\centering
%\setlength{\unitlength}{1.8mm}
%\begin{picture}(24,24)(-10,-10)
%\put(0,0){\vector(1,0){27}}
%\put(0,0){\vector(0,1){11}}
%\put(0,0){\line(-1,0){27}}
%\put(0,0){\line(0,-1){10}}
%\put(12,-0.5){{\line(0,1){1}}}
%\put(-10,-0.5){{\line(0,1){1}}}
%\put(-23,-0.5){{\line(0,1){1}}}
%% \put(-1.2,-1.5){\scriptsize $O$}
%\thicklines
%%\linethickness{0.2mm}
%\put(0,0){{\color{blue}\line(1,0){12}}}
%\put(12,0){{\color{blue} \line(11,9){10}}}
%\put(0,0){{\color{blue}\line(-1,0){10}}}
%\put(-10,0){{\color{blue} \line(-10,-9){10}}}
%\put(-8,12){\scriptsize adjusted realised return $\pp_B(R_T)$}
%\put(20,-2){\scriptsize realised return $R_T$}
%\put(-24,-2){\scriptsize $-1$}
%\put(-10,-2){\scriptsize $l_1$}
%\put(11,-2){\scriptsize $g_1$}
%\put(-7,1){\scriptsize $p_1=0$}
%\put(-17.5,-4.5){\scriptsize $p_2$}
%\put(4.5,1){\scriptsize $f_1=0$}
%\put(18,3.7){\scriptsize $f_2$}
%\end{picture}
%\caption{Provider's adjusted return for a buffer EPS} \label{buffer EPS}
%\end{figure}

In the {\it floor EPS} we set $p_2=0,p_1\in (0,1],f_1=0,f_2\in(0,1]$, which means that the provider covers the holder's losses above the level $l_1$ but the holder's losses below $l_1$ are not covered at all. Notice that, by convention, the buffer is applied to the fee leg and thus it is exactly the same as in a buffer EPS.

\begin{definition} \label{def2.3}
{\rm The \textit{floor EPS} is specified by $p_2=0,p_1\in (0,1],f_1=0,f_2>0$ and thus the cash flow at maturity $T$
to its provider equals
\begin{align*}
\pp_F(R_T)&:=p_1 l_1\I_{\{R_T \in (-1,l_1]\}}-p_1(-R_T)^+\I_{\{R_T\in(l_1,0]\}}+f_2(R_T-g_1)^+\\
&=-p_1(-R_T)^++p_1(l_1-R_T)^++f_2(R_T-g_1)^+ \\
&=-\frac{p_1}{X_0}\,\PutT (X,K^l_1(X_0),T)+\frac{p_1}{X_0}\,\PutT(X,X_0,T)- \frac{f_2}{X_0}\,\CallT (X,K^g_1(X_0),T).
\end{align*}}
\end{definition}

%\begin{figure}\centering
%\setlength{\unitlength}{1.8mm}
%\begin{picture}(24,22)(-10,-10)
%\put(0,0){\vector(1,0){27}}
%\put(0,0){\vector(0,1){11}}
%\put(0,0){\line(-1,0){27}}
%\put(0,0){\line(0,-1){10}}
%\put(12,-0.5){{\line(0,1){1}}}
%\put(-10,-0.5){{\line(0,1){1}}}
%\put(-23,-0.5){{\line(0,1){1}}}
%% \put(0.2,-1.5){\scriptsize $O$}
%%\linethickness{0.6mm}
%\thicklines
%\put(0,0){{\color{blue}\line(1,0){12}}}
%\put(12,0){{\color{blue}\line(11,9){12}}}
%\put(0,0){{\color{blue}\line(-10,-10){10}}}
%\put(-10,-10){{\color{blue}\line(-1,0){13}}}
%\put(-8,12){\scriptsize adjusted realised return $\pp_F(R_T)$}
%\put(20,-2){\scriptsize realised return $R_T$}
%\put(-24,-2){\scriptsize $-1$}
%\put(-10,-2){\scriptsize $l_1$}
%\put(11,-2){\scriptsize $g_1$}
%\put(-7,-4.3){\scriptsize $p_1$}
%%\put(-5,-6.5){\scriptsize $p_1$}
%\put(-18.5,-9){\scriptsize $p_2=0$}
%\put(4.5,1){\scriptsize $f_1=0$}
%\put(18,3.7){\scriptsize $f_2$}
%\end{picture}
%\caption{Provider's adjusted return for a floor EPS}\label{floor EPS}
%\end{figure}

%%%%%%%%%%%%%%%%%%%%%%%%%%%%%%%%%%%%%%%%%%%%%%%%%%%%%%%%%%%%%%%%%%%%%%%%%
\subsection{Static Hedging of a Generic EPS}   \label{sec2.2}
%%%%%%%%%%%%%%%%%%%%%%%%%%%%%%%%%%%%%%%%%%%%%%%%%%%%%%%%%%%%%%%%%%%%%%%%%

After representing the payoff of a generic EPS, one can observe that its payoff can be written in terms of a combination of European calls and puts on rate of return $R_T$ maturing at $T$. Hence our next goal is to examine a hedging portfolio for a generic EPS in terms of options written on the underlying portfolio $X$.

\begin{definition} \label{def2.4}
{\rm A {\it static portfolio} of European put and call options is represented by a vector $h=(a_0,a_1,a_2,\dots ,a_n,b_0,b_1,b_2,\dots ,b_m)\in \mathbb{R}^{n+m+2}$ and its value process $H$ is given by, for every $t\in [0,T]$,
\begin{equation*}
H_t:=\sum_{i=0}^{n}a_i\Putt(X,K^l_i(X_0),T)+\sum_{j=0}^{m}b_j\,\Callt(X,K^g_j(X_0),T)
\end{equation*}
where for $t=T$ we have that $\text{\rm Put}_{\,T}(X,K^l_i(X_0),T)=(K^l_i(X_0)-X_T)^+$ where the strike equals $K^l_i(X_0)=(1+l_i)X_0$ and $\text{\rm Call}_{\,T}(X,K^g_j(X_0),T)=(X_T-K^g_j(X_0))^+$ where $K^g_j(X_0)=(1+g_j)X_0$.}
\end{definition}

According to our convention, the provider receives from the buyer the initial premium $c\in \mathbb{R}$ at time 0 and the cash flow $\pp (R_T)$ at time $T$. Assume that the provider establishes at time 0 a portfolio $h$ of plain-vanilla call and put options written on $X$. Then we have the following definition for the provider's hedged cash flow and fair price
where $B(0,T)$ is the price of the domestic zero-coupon bond with unit notional principal. % $r$ denotes the short term rate.

\begin{definition} \label{def2.5}
{\rm For a generic EPS introduced in Definition \ref{def2.1}, the provider's \textit{hedged cash flow} at time $T$ per one unit of the notional principal $\Np$, denoted by $\CCF_T(c,h)$, equals
\begin{equation*}
\CCF_T(c,h)=(c-H_0) (B(0,T))^{-1}+H_T+\pp (R_T).
\end{equation*}
%\begin{equation*}
%\CCF_T(c,h)=(c-H_0)e^{rT}+H_T+\pp (R_T).
%\end{equation*}
We say that an  EPS can be \textit{statically hedged} by the provider if there exists a static portfolio $\whh$ established at time 0 and an initial premium $\whc \in \mathbb{R}$ such the equality $\CCF_T(\whc,\whh)=0$ holds almost surely. Then the number $\whc$ is called the \textit{fair premium} for an EPS per one unit of the notional principal.}
\end{definition}

Notice that it suffices to find $\whh$ such that the  $\whH_T=N \pp (R_T)$ where $N$ is the notional principal of an EPS
and set $\whc=\whH_0$. The static hedge and fair price for a generic EPS with arbitrary payoff structure can be identified from the following proposition, which is due to Xu et al. \cite{XLR2024} (see Proposition 3.1 therein). Recall that $l_0=g_0=0$ and, for convenience, we also adopt the convention that $f_0=p_0=0$. We stress that the static hedge from Proposition \ref{pro2.1} is model-free and can be easily implemented if the relevant European options are traded.

\begin{proposition} \label{pro2.1}
The static hedge $(\whc,\whh)$ composed of call and put options for a generic EPS with the notional principal $N$ satisfies the following equation
\begin{equation} \label{eq1}
\whH_T= N \sum_{i=0}^{n} \frac{p_{i+1}-p_i}{X_0}\,\PutT (X,K^l_i(X_0),T)
- N \sum_{j=0}^{m} \frac{f_{j+1}-f_j}{X_0}\,\CallT (X,K^g_j(X_0),T)
\end{equation}
where $K^l_i(X_0)=(1+l_i)X_0$ and $K^g_j(X_0)=(1+g_j)X_0$ for every $i=0,1,\dots, n$ and $j=0,1,\dots ,m$. The fair premium $\whc$ for an EPS is equal to $\whH_0$ per one unit of the notional principal where, in the domestic currency,
\begin{equation} \label{eq2}
\whH_0 = \sum_{i=0}^{n} \frac{p_{i+1}-p_i}{X_0}\,\PutO (X,K^l_i(X_0),T)
- \sum_{j=0}^{m} \frac{f_{j+1}-f_j}{X_0}\,\CallO (X,K^g_j(X_0),T).
\end{equation}
\end{proposition}

The proof of Proposition \ref{pro2.1} hinges on the following elementary equalities, for every $i=0,1,\dots,n$,
\begin{align*}
(l_i-R_T)^+ = \frac{((1+l_i)X_0-X_T)^+}{X_0} = \frac{1}{X_0}\,\PutT (X,(1+l_i)X_0,T)
\end{align*}
and, for every $j=0,1,\dots,m$,
\begin{align*}
(R_T-g_j)^+ =\frac{(X_T-(1+g_j)X_0)^+}{X_0}=\frac{1}{X_0}\,\CallT (X,(1+g_j)X_0,T).
\end{align*}

\begin{remark} \label{rem2.1}
{\rm Let us define $\wh{X}_t:=(X_0)^{-1} X_t$ so that  $\wh{X}_0=1$ and the rates of return on processes $\wh{X}$ and $X$ manifestly coincide. Then we obtain an equivalent representation for equality \eqref{eq1}
\begin{equation} \label{eq3}
\whH_T=N\sum_{i=0}^{n}(p_{i+1}-p_i)\,\PutT (\wh{X},1+l_i,T)-N \sum_{j=0}^{m}(f_{j+1}-f_j)\,\CallT (\wh{X},1+g_j,T),
\end{equation}
which shows that, formally, it suffices to examine options written on the normalized process $\wh{X}$. Since in Section \ref{sec3}, the process $X$ is assumed to represent a market quantity (e.g., the price of a given equity or
a specific domestic or foreign market index) we will use there representation \eqref{eq1}, which is based on traded options written on $X$. In contrast, in Section \ref{sec4} we prefer to employ equality \eqref{eq3} since, firstly, options on a bespoke cross-currency portfolio $X$ are unlikely to be traded and, secondly, equation \eqref{eq3} yields simpler pricing and hedging formulae within the cross-currency market model of Section \ref{sec3.2}. To sum up, equality \eqref{eq1} is used when the emphasis is put on the possibility of static hedging of an EPS with traded options on market indices, whereas representation \eqref{eq3} is convenient when performing computations within a stochastic market model, especially when dealing with a basket option on a bespoke portfolio composed of domestic and foreign assets.}
\end{remark}

\begin{remark} \label{rem2.2}
{\rm Recall from Definition \ref{def2.1a} that an EPS is said to be {\it fair} if its initial premium vanishes, that is, $\whc =0$ in Definition \ref{def2.5}. In the case of a buffer or floor EPS, assuming that all parameters except for the fee rate $f_2$ are fixed, one can find a unique level of $f_2$ for which the fair premium for an EPS is null by solving for $f_2$ the linear equation \eqref{eq2} (or, equivalently, \eqref{eq3}) with initial hedging cost $\whH_0$ set to 0. For instance, the {\it fair fee rate} for the buffer EPS equals
\begin{align*}
 \wh{f}_2 = p_2\,\frac{\PutO (X,K^l_1(X_0),T)}{\CallO (X,K^g_1(X_0),T)}= p_2\,\frac{\PutO (\wh{X},1+l_1,T)}{\CallO (\wh{X},1+g_1,T)}
\end{align*}
and for the floor EPS we have}
\begin{align*}
 \wh{f}_2 =p_1\,\frac{\PutO (X,K^l_1(X_0),T)-\PutO (X,X_0,T)}{\CallO (X,K^g_1(X_0),T)}=
  p_1\, \frac{\PutO (\wh{X},1+l_1,T)-\PutO (\wh{X},1,T)}{\CallO (\wh{X},1+g_1,T)}.
\end{align*}
%If more than one nonzero fee rate appears in the specification of an EPS, then the question of its fairness may have several %solutions but it is always easy to structure an EPS in such a way that its fairness is ensured.}
\end{remark}

%%%%%%%%%%%%%%%%%%%%%%%%%%%%%%%%%%%%%%%%%%%%%%%%%%%%%%%%%%%%%%%%%%%%%%%%%%%%%%%%%%%%%%%
%%%%%%%%%%%%%%%%%%%%%%%%%%%%%%%%%%%%%%%%%%%%%%%%%%%%%%%%%%%%%%%%%%%%%%%%%%%%%%%%%%%%%%%
\section{Protection Against Separate Domestic and Foreign Losses}   \label{sec3}
%%%%%%%%%%%%%%%%%%%%%%%%%%%%%%%%%%%%%%%%%%%%%%%%%%%%%%%%%%%%%%%%%%%%%%%%%%%%%%%%%%%%%%%
%%%%%%%%%%%%%%%%%%%%%%%%%%%%%%%%%%%%%%%%%%%%%%%%%%%%%%%%%%%%%%%%%%%%%%%%%%%%%%%%%%%%%%%

For simplicity of notation, we will focus here on the case of two financial markets, the {\it domestic market} in Australia and the foreign market in the U.S., with AUD and USD being the respective currencies, although an extension to several economies does not present any difficulties. Then it suffices to examine two separate hedging portfolios, the first one for the domestic component of holder's portfolio with the notional principal $\Nd$ given in AUD and the second one for the foreign component with the notional principal $\Nf$ given in USD (or, equivalently, with the domestic value $\Nfd = \X_0 \Nf$ in AUD). The exchange rate process $\X$ represents the {\it direct quote}, which means that $\X_t$ gives the price at time $t$ of one unit of the foreign currency (by our convention, USD) in terms of variable number of units of the domestic currency (AUD).

Generally speaking, investors are faced with the choice: either to obtain a protection against losses from their domestic and foreign investments using separate EPSs, one EPS for each component of their portfolio, or to enter into a single EPS furnishing an effective protection against total losses from their combined portfolio of domestic and foreign assets.
In the latter case, the portfolio's value process is expressed in the domestic currency using a particular accounting method.

Assume first that the investor decides to purchase two independent products yielding protection against negative returns on either $\Sd$ or $\Sf$, respectively, where $\Sd_t$ and $\Sf_t$ represent the current values at time $t$ of domestic and foreign reference portfolios (e.g., particular market indices) with the respective values expressed in the domestic and foreign currency.  To be more specific, we first assume that the payoff from domestic EPS with the profile $\ppd$ (resp. the foreign EPS with the profile $\ppf$) is referencing the domestic return $\Rd_T$ (resp. the {\it nominal} foreign return $\Rfd_T$) with the corresponding domestic and foreign EPSs examined in Sections \ref{sec3.3} and \ref{sec3.4}, respectively.  Next, we introduce and study the domestic pricing of foreign EPSs based on the {\it effective} and {\it quanto} foreign returns in Sections \ref{sec3.5} and \ref{sec3.6}, respectively. Subsequently, in Section \ref{sec4}, we study cross-currency EPSs that yield protection against combined effective losses of a cross-currency portfolio.

%%%%%%%%%%%%%%%%%%%%%%%%%%%%%%%%%%%%%%%%%%%%%%%%%%%%%%%%%%%%%%%%%%%%%%%%%%%%%%%
\subsection{Cross-Currency Portfolios}  \label{sec3.1}
%%%%%%%%%%%%%%%%%%%%%%%%%%%%%%%%%%%%%%%%%%%%%%%%%%%%%%%%%%%%%%%%%%%%%%%%%%%%%%%

We introduce the following notation for the {\it domestic rate of return} and the (nominal) {\it foreign rate of return}, as well as the return on the foreign currency, for every $t\in [0,T]$,
\begin{equation*}
\Rd_t:=\frac{\Sd_t-\Sd_0}{\Sd_0},\quad \Rf_t=\frac{\Sf_t-\Sf_0}{\Sf_0}, \quad \RQ_t=\frac{\X_t-\X_0}{\X_0}.
\end{equation*}
Another important concept is the {\it effective foreign rate of return,} that is, the rate of return on the foreign portfolio $\Sf$ when its value at any date $t$ is expressed in the domestic currency using the current value $X_t$ of the exchange rate. The effective foreign return satisfies, for every $t\in [0,T]$,
\begin{equation*}
\Rfd_t := \frac{\Sfd_t-\Sfd_0}{\Sfd_0}=\frac{\X_t \Sf_t -\X_0 \Sf_0}{\X_0 \Sf_0} = \Rf_t + \RQ_t + \Rf_t \RQ_t
\end{equation*}
where the quantity $\Sfd_t :=\X_t\Sf_t$ represents the current domestic value of $\Sf_t$.

We also need the notation for the notional principal of the holder's cross-currency portfolio and for its decomposition into the domestic and foreign holdings. We consider an agent who invests at time 0 in a portfolio of domestic and foreign risky assets with a total wealth equal to $\Np$ (in AUD). Then the notional principal $\Np$ can be decomposed
into domestic and foreign holdings as follows $\Np=\Nd+\Nfd$ where $\Nd$ (in AUD) is the notional principal of domestic assets in holder's portfolio and $\Nfd$ (also given in AUD) is the notional principal of foreign assets expressed in the domestic currency. Recall that we denote by $\Nf$ the notional principal of foreign assets in holder's portfolio expressed in the foreign currency, which means that $\Nf =(\X_0)^{-1}\Nfd$ (in USD). This leads to an equivalent representation $\Np=\Nd+\X_0\Nf$.

It is convenient to introduce the weight $w:=\Nd/\Np$ (resp. $1-w:=\Nfd/\Np$), which gives the proportion of the total wealth invested at initial date in domestic assets (resp. foreign assets). Then the number of units of the domestic and foreign reference portfolios $\Sd$ and $\Sf$ held in a cross-currency portfolio satisfy
\begin{equation*}
\Nd= w\Np=\alpha_0 \Sd_0, \quad \Nfd=(1-w)\Np=\beta_0 \Sfd_0
\end{equation*}
where $\alpha_0$ and $\beta_0$ represent holdings in $\Sd$ and $\Sf$, respectively, and satisfy
\begin{equation*}
\alpha_0:=\frac{\Nd}{\Sd_0}=\frac{w\Np}{\Sd_0}, \quad \beta_0:=\frac{\Nfd}{\Sfd_0}=\frac{(1-w)\Np}{\Sfd_0}= \frac{(1-w)\Np}{\X_0\Sf_0}.
\end{equation*}
Then the {\it cross-currency portfolio} in domestic assets represented by $\Sd$ and foreign assets given by $\Sf$ with the notional principal $\Np$ and the weight $w$ can be identified with the pair $(\alpha_0,\beta_0)$.

%%%%%%%%%%%%%%%%%%%%%%%%%%%%%%%%%%%%%%%%%%%%%%%%%%%%%%%%%%%%%%%%%%%%%%%%%%%%%%%%%%%%%
\subsection{Cross-Currency Market Model}         \label{sec3.2}
%%%%%%%%%%%%%%%%%%%%%%%%%%%%%%%%%%%%%%%%%%%%%%%%%%%%%%%%%%%%%%%%%%%%%%%%%%%%%%%%%%%%%

Although our main focus is on a model-free static hedging of an EPS, we will also use a particular stochastic model to give an
alternative pricing method, which will be used in assessment of the forward performance of various classes of
cross-currency EPSs. Therefore, we briefly summarize in this section the classical cross-currency model as presented, for instance, in Chapter 13 of Musiela and Rutkowski \cite{MM2002}. The domestic and foreign bank accounts are given by $dB_t^d = r^d B_t^d\,dt$ and $dB_t^f = r^f B_t^f\,dt$ for some constant short-term interest rates $r^d$ and $r^f$. For a fixed horizon date $T>0$, the dynamics of domestic and foreign equities, $\Sd$ and $\Sf$, and the exchange rate, $\X$, under the physical probability measure $\bbP$ are specified by the SDEs
\begin{align} \label{eq4a}
d\Sd_t=\Sd_t\big(\mu_t^d\,dt+\bm{\sigma}^d\,d\textbf{W}_t\big),\nonumber \\
d\Sf_t=\Sf_t\big(\mu_t^f\,dt+\bm{\sigma}^f\,d\textbf{W}_t\big),\\
d\X_t=\X_t\big(\mu_t^q\,dt+\bm{\sigma}^q\,d\textbf{W}_t\big), \nonumber
\end{align}
with strictly positive initial conditions $\Sd_0,\Sf_0$ and $\X_0$ where the column vector $\textbf{W}= [W^1,W^2,W^3]^*$ where $^*$ denotes the transpose is a standard Brownian motion on $(\Omega ,\mathcal{F},\bbP)$ and the volatilities are given by row vectors: $\bm{\sigma}^d=[\sigma^{1d},\sigma^{2d},\sigma^{3d}],\, \bm{\sigma}^f = [\sigma^{1f}, \sigma^{2f}, \sigma^{3f}]$
and $\bm{\sigma}^q=[\sigma^{1q},\sigma^{2q},\sigma^{3q}]$. Then the dynamics of the domestic and foreign equity
and the exchange rate under an equivalent probability measure $\Q^d$, which is the {\it domestic martingale measure}, are given as
\begin{align} \label{eq4}
&d\Sd_t=\Sd_t\big( (r^d-\kappa^d)\,dt+\bm{\sigma}^d\,d\textbf{W}_t^d\big),\nonumber \\
&d\Sf_t=\Sf_t \big((r^f-\kappa^f-\bm{\sigma}^f(\bm{\sigma}^q)^*)\,dt+\bm{\sigma}^f\,d\textbf{W}_t^d\big),\\
&d\X_t=\X_t\big((r^d-r^f)\,dt+\bm{\sigma}^q\,d\textbf{W}_t^d\big), \nonumber
\end{align}
where $\textbf{W}^d$ is a standard Brownian motion on $(\Omega ,\mathcal{F},\Q^d)$. Finally, the process $\textbf{W}_t^f=\textbf{W}_t^d-(\bm{\sigma}^q)^*\,t$ is a standard Brownian motion
on $(\Omega ,\mathcal{F},\Q^f)$ where $\Q^f$ is the {\it foreign martingale measure} and thus the processes
$\Sd,\Sf$ and $\X$ are governed under $\Q^f$ by the following SDEs
\begin{align}  \label{eq5}
&d\Sd_t=\Sd_t\big((r^d-\kappa^d+\bm{\sigma}^d (\bm{\sigma}^q)^*)\,dt+\bm{\sigma}^d\,d\textbf{W}_t^f\big),\nonumber\\
&d\Sf_t=\Sf_t\big((r^f-\kappa^f)\,dt+\bm{\sigma}^f \,d\textbf{W}_t^f\big),\\
&d\X_t=\X_t\big((r^d-r^f+\bm{\sigma}^q (\bm{\sigma}^q)^*)\,dt+\bm{\sigma}^q\,d\textbf{W}_t^f\big).\nonumber
\end{align}
It is well known that if the volatility matrix $\Sigma = [\bm{\sigma}^{d*},\bm{\sigma}^{f*},\bm{\sigma}^{q*}]$ is non-singular, then the cross-currency market model given by \eqref{eq4} is arbitrage-free and complete.

%The model parameters used in our numerical studies presented in Section \ref{sec5} and examples are reported in Table \ref{table1}. The prices of plain-vanilla domestic or foreign options are given by the Black-Scholes formula and the pricing formulae for other options of our interest are given in Propositions \ref{pro3.1}, \ref{pro3.2}, \ref{pro4.1}, \ref{pro4.2}, \ref{pro4.3}, \ref{pro4.4}, \ref{pro5.1}, \ref{pro5.2}, \ref{pro5.3} and \ref{pro5.4}.

It is worth noting that initial prices of the foreign and domestic equity is immaterial when pricing
equity protection swaps within the setup of a cross-currency market model. This is easy to understand since
the payoffs of equity protection swaps only depend on realized returns on the reference domestic and foreign portfolios,
and respective notional principals so they are independent of the initial prices of equities.
However, when presenting implementations of static hedging for various classes of EPSs, we assume that
the trading strategies are based on actively traded options and thus the initial market prices of equities
and actual strikes for traded options will be used. To sum up, one should note the crucial distinction between the real-world pricing results based on market quotes for traded options, which give {\it market prices} of EPSs, and theoretical computations yielding the {\it model prices} of options obtained within a cross-currency market model.
The former approach relies on the existence of a liquid market for relevant European options, whereas the latter hinges on an implicit assumption that European options can be dynamically hedged using underlying domestic and foreign assets.

%%%%%%%%%%%%%%%%%%%%%%%%%%%%%%%%%%%%%%%%%%%%%%%%%%%%%%%%%%%%%%%%%%%%%%%%%%%%%%%%%%%%%
\subsection{Protection Against Domestic Losses}     \label{sec3.3}
%%%%%%%%%%%%%%%%%%%%%%%%%%%%%%%%%%%%%%%%%%%%%%%%%%%%%%%%%%%%%%%%%%%%%%%%%%%%%%%%%%%%%

To protect the domestic component of his portfolio with the initial wealth $\Np$ and weight $w$, the holder may use
the {\it domestic EPS} with the notional principal $\Nd=w\Np$ and the payoff at time $T$ given in the domestic currency
\begin{equation*}
\EPSd:=\Nd \ppd (\Rd_T)=w \Np \ppd (\Rd_T).
\end{equation*}
We write $K^d$ to denote the strike price denominated in the domestic currency (resp., the foreign currency).
It is assumed that call and put options with payoffs $C_T(\Sd,K^d,T)=(\Sd_T-K^d)^+$ and $P_T(\Sd,K^d,T)=(K^d-\Sd_T)^+$ for
all strikes $K^d$ are traded, which is a natural assumption given that the portfolio $\Sd$ typically represents a market index. The case of the domestic EPS does not demand a detailed study since it was already examined in Xu et al. \cite{XLR2024}. We stress that all results on static hedging are model-free and thus they do not rely on the cross-currency market model
introduced in Section \ref{sec3.2}. The following result is an easy consequence of Proposition \ref{pro2.1}.

\begin{corollary} \label{cor3.1}
The static hedging in terms of domestic equity options for the domestic EPS is given by Proposition \ref{pro2.1} with
$\pp=\ppd,N=\Nd$ and $X=\Sd$, that is, $\EPSd=\whHd$ where for every $t\in [0,T]$,
\begin{equation*}
\whHdt=\Nd\sum_{i=0}^{n}\frac{p_{i+1}-p_i}{\Sd_0}\,P_t(\Sd,K^l_i(\Sd_0),T)
 -\Nd\sum_{j=0}^{m}\frac{f_{j+1}-f_j}{\Sd_0}\,C_t(\Sd,K^g_j(\Sd_0),T)
\end{equation*}
where $K^l_i(\Sd_0)=(1+l_i)\Sd_0$ and $K^g_j(\Sd_0)=(1+g_j)\Sd_0$ for every $i=0,1,\dots, n$ and $j=0,1,\dots ,m$.
\end{corollary}

Typically, the domestic and foreign portfolios of investors can be associated with existing market indices, such as the
S\&P 500 Index in the U.S. or the S\&P/ASX 200 Index in Australia, which encompass a large majority of listed equities in respective economy. Furthermore, many financial institutions maintain and offer shares in exchange traded funds (ETFs) containing multiple stocks which are designed to meet a predetermined goal,  most frequently, to approximate the performance of a certain market index. Hence ETFs can be considered as actively traded proxies for market indices and thus the assumption that market indices can be seen as traded assets has some merit.

Furthermore, it is clear from Corollary \ref{cor3.1} that liquidly traded options on market indices, which are available on stock exchanges worldwide, can be used by providers of domestic and foreign EPSs for implementation of static hedging and hence also pricing of various EPS products they offer, not only when they are issued but also for marking-to-market if, for instance, the cancellation of an EPS by its holder is allowed. In this context, it is worth mentioning the CBOE S\&P 500 Index options (SPX), Nanos S\&P 500 Index options (NANOS), and SPDR S\&P 500 ETF options (SPY) offered in the U.S., as well as ASX S\&P/ASX 200 Index options (XJO) actively traded in Australia. The above mentioned index and ETF options are of European style, except for SPY options, which are of American style.
% Link: https://www.cboe.com/tradable_products/sp_500/spx_options/
% Link: https://www.asx.com.au/documents/resources/index_options.pdf

Plain-vanilla options on most typical domestic and foreign portfolios, $\Sd$ and $\Sf$, are likely to be actively traded on options exchanges and thus Corollary \ref{cor3.1}, which is a model-free result, is sufficient for real-life market pricing and static hedging of a domestic EPS. In contrast, Corollary \ref{cor3.2} is useful when studying model pricing of domestic EPSs and using their model prices in the context of more complex cross-currency EPSs, which are studied in the foregoing sections. If we assume that the dynamics of the domestic portfolio $\Sd$ are given by \eqref{eq4}, then we may use the classical Black-Scholes options pricing formula to obtain the following pricing result for the domestic EPS where we write $\Phin$ to denote the cumulative density function of the standard normal distribution.

\begin{corollary} \label{cor3.2}
The model price at time 0 of the domestic EPS with the payoff $\EPSd$ at time $T$ equals, in the domestic currency,
\begin{align*}
\EPSdo&=\Nd\sum_{i=0}^{n}(p_{i+1}-p_i)\Big( e^{-r^d T}(1+l_i) \Phin\big(-d_-(l_i,T)\big)-e^{-\kappa^d T}\Phin\big(-d_+(l_i,T)\big)\Big)\\
&-\Nd\sum_{j=0}^{m}(f_{j+1}-f_j)\Big(e^{-\kappa^d T}\Phin\big(d_+(g_j,T)\big)- e^{-r^d T}(1+g_j)\Phin\big(d_-(g_j,T)\big)\Big)
\end{align*}
where
\begin{equation*}
d_{\pm}(x,T):=\frac{1}{\|\bm{\sigma}^d\|\sqrt{T}}\bigg[-\ln (1+x)+ \Big(r^d -\kappa^d \pm\frac{1}{2}\|\bm{\sigma}^d\|^2\Big) T\bigg].
\end{equation*}
\end{corollary}

%%%%%%%%%%%%%%%%%%%%%%%%%%%%%%%%%%%%%%%%%%%%%%%%%%%%%%%%%%%%%%%%%%%%%%%%%%%%%%%%%%%%%
\subsection{Protection Against Nominal Foreign Losses}    \label{sec3.4}
%%%%%%%%%%%%%%%%%%%%%%%%%%%%%%%%%%%%%%%%%%%%%%%%%%%%%%%%%%%%%%%%%%%%%%%%%%%%%%%%%%%%%

Analogous pricing results can be formulated for the {\it foreign} EPS with the payoff denominated in the foreign currency and given by $\EPSf:=\Nf\ppf(\Rf_T)$ where the payoff function $\ppf$ is given by parameters $\wh{p}_i,\wh{f}_j,\wh{l}_i$ and $\wh{g}_j$. Then the model-free pricing formula for the foreign EPS with the payoff $\ppf (\Rf_T)$ at maturity $T$ and the foreign notional principal $\Nf$ reads, in the foreign currency,
\begin{equation*}
\EPSf=\Nf \sum_{i=0}^{\wh{n}}\frac{\wh{p}_{i+1}-\wh{p}_i}{\Sf_0}\,P_T(\Sf,K^{\wh{l}}_i(\Sf_0),T)
 - \Nf \sum_{j=0}^{\wh{m}}\frac{\wh{f}_{j+1}-\wh{f}_j}{\Sf_0}\,C_T(\Sf,K^{\wh{g}}_j(\Sf_0),T)
\end{equation*}
where $K^{\wh{l}}_i(\Sd_0)=(1+\wh{l}_i)\Sd_0$ and $K^{\wh{g}}_j(\Sd_0)=(1+\wh{g}_j)\Sd_0$ for every $i=0,1,\dots, \wh{n}$ and $j=0,1,\dots ,\wh{m}$. Using dynamics \eqref{eq5} under the foreign martingale measure $\Q^f$, we obtain the foreign market counterpart of Proposition \ref{pro3.2}.

\begin{corollary} \label{cor3.2a}
The model price at time 0 of the foreign EPS with the foreign notional principal $\Nf$ equals, in the foreign currency,
\begin{align*}
\EPSfo&=\Nf \sum_{i=0}^{\wh{n}}(\wh{p}_{i+1}-\wh{p}_i)\Big( e^{-r^f T}(1+\wh{l}_i) \Phin\big(-d_-(\wh{l}_i,T)\big)
    -e^{-\kappa^f T}\Phin\big(-d_+(\wh{l}_i,T)\big)\Big)\\
&-\Nf \sum_{j=0}^{\wh{m}}(\wh{f}_{j+1}-\wh{f}_j)\Big(e^{-\kappa^f T}\Phin\big(d_+(\wh{g}_j,T)\big)- e^{-r^f T}(1+\wh{g}_j)\Phin\big(d_-(\wh{g}_j,T)\big)\Big)
\end{align*}
where
\begin{equation*}
d_{\pm}(x,T):=\frac{1}{\|\bm{\sigma}^f\|\sqrt{T}}\bigg[-\ln (1+x)+ \Big(r^f -\kappa^f \pm\frac{1}{2}\|\bm{\sigma}^f\|^2\Big) T\bigg].
\end{equation*}
\end{corollary}

%%%%%%%%%%%%%%%%%%%%%%%%%%%%%%%%%%%%%%%%%%%%%%%%%%%%%%%%%%%%%%%%%%%%%%%%%%%%%%%%%%%%%
\subsection{Protection Against Effective Foreign Losses}    \label{sec3.5}
%%%%%%%%%%%%%%%%%%%%%%%%%%%%%%%%%%%%%%%%%%%%%%%%%%%%%%%%%%%%%%%%%%%%%%%%%%%%%%%%%%%%%

A more practical protection against losses from the foreign part of a cross-currency portfolio with the  notional principal $\Np$ and the weight $1-w$ for foreign assets can be obtained by entering into the {\it effective foreign EPS} with the notional principal $\Nfd=\X_0 \Nf $ and with the payoff at time $T$ to the provider given by
\begin{equation*}
\EPSfe:=\Nfd\ppf(\Rfd_T)=\X_0 \Nf \ppf(\Rfd_T)=(1-w)\Np\ppf(\Rfd_T).
\end{equation*}
Consider a foreign equity call option struck in the domestic currency where the holder has the right to buy at time $T$ one unit of $\Sf$ by paying $K^d$ units of domestic currency, which yields the payoff equal to $(\Sfd_T-K^d)^+$ in the domestic currency. We set  $\wtC_T(\Sfd,K^d,T)=(\Sfd_T-K^d)^+$ and $\wtP_T(\Sfd,K^d,T)=(K^d-\Sfd_T)^+$ and we denote by
$\wtC(\Sfd,K^d,T)$ and $\wtP(\Sfd,K^d,T)$ the price processes of call and put options written on the
synthetic risky asset $\Sfd = Q \Sf$ representing the domestic value of a reference portfolio of foreign assets.

\begin{corollary} \label{cor3.3}
The static hedging in terms of foreign equity call and put options struck in the domestic currency for an effective
foreign EPS is given by Proposition \ref{pro2.1} with $\pp=\ppf,N=\Nfd$ and $X=\Sfd$, that is,
$\EPSfe=\whHfe$ where, for every $t\in [0,T]$,
\begin{equation*}
\whHfet=\Nfd\sum_{i=0}^{\wh{n}}\frac{\wh{p}_{i+1}-\wh{p}_i}{\Sfd_0}\,\wtP_t(\Sfd,K^{\wh{l}}_i(\Sfd_0),T)
-\Nfd\sum_{j=0}^{\wh{m}}\frac{\wh{f}_{j+1}-\wh{f}_j}{\Sfd_0}\,\wtC_t(\Sfd,K^{\wh{g}}_j(\Sfd_0),T)
 \end{equation*}
where $K^{\wh{l}}_i(\Sfd_0)=(1+\wh{l}_i)\Sfd_0$ and $K^{\wh{g}}_j(\Sfd_0)=(1+\wh{g}_j)\Sfd_0$ for every $i=0,1,\dots,\wh{n}$ and $j=0,1,\dots,\wh{m}$.
\end{corollary}

% \begin{remark} \label{rem3.2} {\rm
For the purpose of static hedging of an effective foreign EPS, we assume that the provider can trade in foreign equity options with the foreign equity struck in domestic currency as the underlying asset and strike price given in the domestic currency.
Several ETFs available in Australia allow to track the performance of the S\&P 500 AUD Index, which is the index designed to measure the AUD performance of large capitalisation U.S. equities. Investing in such ETFs can be seen as a proxy for the theoretical concept of the value process $\Sfd$ representing the domestic value of a foreign portfolio, e.g., a particular market index in the foreign economy. Corollary \ref{cor3.3} shows that related options of European style could be used to hedge and price effective foreign EPSs and thus, in principle, their current market values would be given in terms of market quotes for options on such ETFs provided that they are publicly available. %}
%  Link: https://www.finder.com.au/share-trading/exchange-traded-funds/best-sp-500-etfs
%\end{remark}

It is a realistic assumption that the above mentioned options are available on the OTC market, in which case the provider of an effective foreign EPS can establish the static hedging strategy. However, for the purpose of obtaining a better insight into properties of arbitrage-free prices of effective foreign EPSs of various classes and with differing parameters, we will also use the following classical result, which is known to be valid within the setup of Section \ref{sec3.1}. In fact, it is an easy consequence of the dynamics of $\Sfd$ under $\Qd$
\begin{align*}
d\Sfd_t = \Sfd_t \big( (r^d-\kappa^f) \,dt+(\bm{\sigma}^f+\bm{\sigma}^q)\,d\textbf{W}_t^d\big).
\end{align*}

\begin{proposition} \label{pro3.1}
The model price at time $t$ of a foreign call option struck in the domestic currency with the terminal payoff at maturity $T$ given by $\wtC_T(\Sfd,K^d,T)= (\Sfd_T-K^d)^+$ is equal to, for every $t\in [0,T]$,
\begin{equation*}
\wtC_t(\Sfd,K^d,T)=\wt{S}^{f,e}_t\Phin\big(\wtd_+(\Sfd_t,T-t)\big)-\wtK^d_t\Phin\big(\wtd_-(\Sfd_t,T-t)\big)
\end{equation*}
where $\wt{S}^{f,e}_t:=  e^{-\kappa^f(T-t)}\Sfd_t, \wtK^d_t:= e^{-r^d(T-t)}K^d$ and
\begin{equation*}
\wtd_{\pm}(\Sfd_t,T-t):=\frac{1}{\|\bm{\sigma}^f+\bm{\sigma}^q\|\sqrt{T-t}}
\bigg[\ln \frac{\wt{S}^{f,e}_t}{\wtK^d_t}\pm\frac{1}{2}\|\bm{\sigma}^f+\bm{\sigma}^q\|^2(T-t)\bigg].
\end{equation*}
The price of a put option with the payoff $\wtP_T(\Sfd,K^d,T)=(K^d-\Sfd_T)^+$ equals
\begin{equation*}
\wtP_t(\Sfd,K^d,T)=\wtK^d_t\Phin\big(-\wtd_-(\Sfd_t,T-t)\big)-\wt{S}^{f,e}_t\Phin\big(-\wtd_+(\Sfd_t,T-t)\big).
\end{equation*}
\end{proposition}

Notice that Proposition \ref{pro3.1} allows for a dynamic hedging of an effective foreign EPS using the underlying portfolio
$\Sf$, which is formally represented by the process $\Sfd$, and the domestic money market account $B^d$. Furthermore, by combining Corollary \ref{cor3.3} with Proposition \ref{pro3.1}, we obtain the pricing formula for an effective foreign EPS, which holds within the setup of Section \ref{sec3.2}.

\begin{corollary} \label{cor3.4}
The model price at time 0 of the effective foreign EPS with the payoff $\EPSfe$ at time $T$ equals, in the domestic currency,
\begin{align*}
\EPSfeo&=\Nfd\sum_{i=0}^{\wh{n}}(\wh{p}_{i+1}-\wh{p}_i)\Big(e^{-r^d T}(1+\wh{l}_i) \Phin\big(-k_-(\wh{l}_i,T)\big)-e^{-\kappa^f T}\Phin\big(-k_+(\wh{l}_i,T)\big)\Big)\\
&-\Nfd\sum_{j=0}^{\wh{m}}(\wh{f}_{j+1}-\wh{f}_j)\Big(e^{-\kappa^f T}\Phin\big(k_+(\wh{g}_j,T)\big)
- e^{-r^d T}(1+\wh{g}_j)\Phin\big(k_-(\wh{g}_j,T)\big)\Big)
\end{align*}
where
\begin{equation*}
k_{\pm}(x,T):=\frac{1}{\|\bm{\sigma}^f+\bm{\sigma}^q\|\sqrt{T}}
\bigg[-\ln (1+x)+ \Big(r^d - \kappa^f \pm\frac{1}{2}\|\bm{\sigma}^f+\bm{\sigma}^q\|^2\Big) T\bigg].
\end{equation*}
\end{corollary}

\subsection{Protection Against Quanto Foreign Losses}      \label{sec3.6}
%%%%%%%%%%%%%%%%%%%%%%%%%%%%%%%%%%%%%%%%%%%%%%%%%%%%%%%%%%%%%%%%%%%%%%%%%%%%%%%%%%%%%

A {\it quantity-adjusting derivative} (or, simply, a {\it quanto derivatinve}) is a cross-currency derivative in which the underlying asset is denominated in one currency but settlement is made in another currency at a predetermined exchange rate, hereafter denoted as $\overline{\X}$. Another name for such a financial product is a {\it guaranteed exchange rate derivative} since they enable investors to fix at the contract's inception the exchange rate between two currencies. By construction, quanto derivatives help to safeguard the investors against unfavourable fluctuations in the exchange rate but, of course, at the cost of no benefiting from a favourable move of the exchange rate.

We henceforth assume that the exchange rate $\overline{\X}$ has been fixed at time 0 and at any date $t>0$ the domestic value of foreign holdings is calculated using the exchange rate $\overline{\X}$, which in fact is an investment opportunity offered to investors by some financial companies. In principle, the exchange rate $\overline{\X}$ in a quanto contract is completely arbitrary but we sometimes set $\overline{\X}=\X_0$ to simplify the presentation.

We first consider a quanto EPS with a reference portfolio composed of foreign assets only and with the foreign notional principal $\Nf$. In the {\it quanto foreign} EPS with the predetermined exchange rate $\overline{\X}$, the provider's cash flow at time $T$ equals, in the domestic currency,
\begin{equation*}
\EPSfq:=\overline{\X}\Nf\ppf(\Rf_T). %=\Nfd \ppf(\Rf_T).
\end{equation*}
Similarly, the terminal payoffs of the quanto call and put option written on $\Sf$ are defined by $\overline{C}_T(\Sf,K^f,T,\overline{\X}):=\overline{\X}(\Sf_T-K^f)^+$ and $\overline{P}_T(\Sf,K^f,T,\overline{\X}):=\overline{\X}(K^f-\Sf_T)^+$ where $\Sf_T$ and $K^f$ are given in the foreign currency and the option's payoff is expressed in the domestic currency. Similarly, let $\overline{C}(\Sf,K^f,T,\overline{\X})$ and $\overline{P}(\Sf,K^f,T,\overline{\X})$ stand for the domestic price processes
of quanto call and put options written on the foreign asset $\Sf$.
We first state a model-free result on static hedging of a quanto foreign EPS and hence also the price in terms of quanto options.

\begin{corollary} \label{cor3.5}
The static hedging in terms of quanto options for the quanto foreign EPS is given by Proposition \ref{pro2.1} with $\pp=\ppf,N = \Nf$ and $X=\Sf$, that is, $\EPSfq=\whHfq$ where, for every $t\in [0,T]$,
\begin{equation*}
\whHfqt=\Nf\sum_{i=0}^{\wh{n}}\frac{\wh{p}_{i+1}-\wh{p}_i}{\Sf_0}\,\overline{P}_t(\Sf,K^{\wh{l}}_i(\Sf_0),T,\overline{\X})
-\Nf\sum_{j=0}^{\wh{m}}\frac{\wh{f}_{j+1}-\wh{f}_j}{\Sf_0}\,\overline{C}_t(\Sf,K^{\wh{g}}_j(\Sf_0),T,\overline{\X})
\end{equation*}
where $K^{\wh{l}}_i(\Sf_0)=(1+\wh{l}_i)\Sf_0$ and $K^{\wh{g}}_j(\Sf_0)=(1+\wh{g}_j)\Sf_0$ for every $i=0,1,\dots, \wh{n}$ and $j=0,1,\dots ,\wh{m}$.
\end{corollary}

%\begin{remark} \label{rem3.3} {\rm  Similar to Remark \ref{rem3.2},
As a proxy for the concept of the quanto foreign return, we could mention returns on exchange traded funds aiming to provide investors with the performance of the S\&P 500 AUD Hedged Index, which uses a currency hedged strategy to reduce the risk of currency fluctuations while accepting underlying market risk of major U.S. equities.
As usual, we refer to AUD as the domestic currency but analogous indices are also offered for other major currencies. Then,
as shown in Corollary \ref{cor3.5}, the market value of the quanto foreign EPS would be given in terms of market prices of options on such ETFs, of course, provided that they are traded on the OTC market. %}
%  Link: https://www.blackrock.com/au/individual/products/271027/ishares-s-and-p-500-aud-hedged-etf
%\end{remark}

If the provider of the quanto foreign EPS cannot directly trade in quanto options, then a plausible alternative is to use a stochastic model to compute their prices and obtain a synthetic versions of quanto options by dynamic trading. To this end, one may use the following result, which holds in the model of Section \ref{sec3.1}. Let us denote $\delta^q :=r^f-\kappa^f-r^d -\bm{\sigma}^f(\bm{\sigma^q})^*$. Then we have the following result.

\begin{proposition} \label{pro3.2}
The model price of a quanto call option with the payoff $\overline{C}_T(\Sf,K^f,T,\overline{\X})=\overline{\X}(\Sf_T-K^f)^+$ is given by, in the domestic currency,
\begin{equation*}
\overline{C}_t(\Sf,K^f,T,\overline{\X})=\overline{\X}\Big(\Sf_t e^{\delta^q(T-t)}\Phin\big(\bar{d}_+(\Sf_t,T-t)\big)
-\wtK^f_t \Phin\big(\bar{d}_-(\Sf_t,T-t)\big)\Big)
\end{equation*}
where $\wtK^f_t =e^{-r^d(T-t)}K^f$ and
\begin{equation*}
\bar{d}_{\pm}(\Sf_t,T-t)=\frac{1}{\|\bm{\sigma}^f \|\sqrt{T-t}}
\bigg[\ln \frac{\Sf_t}{\wtK^f_t}+\Big(\delta^q \pm \frac{1}{2}\|\bm{\sigma}^f\|^2\Big)(T-t)\bigg].
\end{equation*}
The price of a quanto put option with the payoff $\overline{P}_T(\Sf,K^f,T,\overline{\X})=\overline{\X}(K^f-\Sf_T)^+$ equals
\begin{equation*}
\overline{P}_t(\Sf,K^f,T,\overline{\X})=\overline{\X} \Big(\wtK^f_t \Phin\big(-\bar{d}_-(\Sf_t,T-t)\big)
-\Sf_t e^{\delta^q(T-t)}\Phin\big(-\bar{d}_+(\Sf_t,T-t)\big)\Big).
\end{equation*}
\end{proposition}

From Corollary \ref{cor3.5} and Proposition \ref{pro3.2}, we obtain the pricing formula for the quanto
foreign EPS in the model of Section \ref{sec3.2}. Notice that in Corollary \ref{cor3.6} we set $\overline{\X}=\X_0$ and thus the equality $\overline{\X}\Nf=\Nfd$ holds.

\begin{corollary} \label{cor3.6}
The model price at time 0 of the quanto foreign EPS with the payoff $\EPSfq$ at time $T$ with $\overline{\X}=\X_0$ equals,
in the domestic currency,
\begin{align*}
\EPSfqo&=\Nfd\sum_{i=0}^{\wh{n}}(\wh{p}_{i+1}-\wh{p}_i)\Big(e^{-r^d T}(1+\wh{l}_i)\Phin\big(-\bar{k}_-(\wh{l}_i,T)\big)
-e^{\delta^q T}\Phin\big(-\bar{k}_+(\wh{l}_i,T)\big)\Big)\\
&-\Nfd \sum_{j=0}^{\wh{m}}(\wh{f}_{j+1}-\wh{f}_j)\Big(e^{\delta^q T}\Phin\big(\bar{k}_+(\wh{g}_j,T)\big)
-e^{-r^d T}(1+\wh{g}_j)\Phin\big(\bar{k}_-(\wh{g}_j,T)\big)\Big)
 \end{align*}
 where
 \begin{equation*}
\bar{k}_{\pm}(x,T)=\frac{1}{\|\bm{\sigma}^f \|\sqrt{T}}
\bigg[-\ln(1+x)+\Big(r^f-\bm{\sigma}^f(\bm{\sigma^q})^*\pm\frac{1}{2}\|\bm{\sigma}^f\|^2\Big)T\bigg].
\end{equation*}
\end{corollary}

%%%%%%%%%%%%%%%%%%%%%%%%%%%%%%%%%%%%%%%%%%%%%%%%%%%%%%%%%%%%%%%%%%%%%%%%%%%%%%%%%%%%%%%%
%%%%%%%%%%%%%%%%%%%%%%%%%%%%%%%%%%%%%%%%%%%%%%%%%%%%%%%%%%%%%%%%%%%%%%%%%%%%%%%%%%%%%%%%
\section{Pricing and Hedging of Cross-Currency Basket EPSs}     \label{sec4}
%%%%%%%%%%%%%%%%%%%%%%%%%%%%%%%%%%%%%%%%%%%%%%%%%%%%%%%%%%%%%%%%%%%%%%%%%%%%%%%%%%%%%%%%
%%%%%%%%%%%%%%%%%%%%%%%%%%%%%%%%%%%%%%%%%%%%%%%%%%%%%%%%%%%%%%%%%%%%%%%%%%%%%%%%%%%%%%%%

It is natural to expect that holders of superannuation accounts will focus on the value of their foreign holdings expressed in the domestic currency and pooled with their domestic holdings. Therefore, after discussing separate EPSs protecting against either domestic or foreign losses, we move on to examination of the case of a cross-currency basket EPS, which refers to aggregated returns of a basket of equity indices. As usual, $\Np$ is the nominal value in the domestic currency of the holder's portfolio of domestic and foreign assets and also the domestic notional principal of a cross-currency basket
EPS.

%%%%%%%%%%%%%%%%%%%%%%%%%%%%%%%%%%%%%%%%%%%%%%%%%%%%%%%%%%%%%%%%%%%%%%%%%%%%%%%%%%%%%
\subsection{Cross-Currency Basket EPS}         \label{sec4.1}
%%%%%%%%%%%%%%%%%%%%%%%%%%%%%%%%%%%%%%%%%%%%%%%%%%%%%%%%%%%%%%%%%%%%%%%%%%%%%%%%%%%%%

By the {\it effective wealth} of a cross-currency portfolio we mean its total value expressed in the domestic currency.
Specifically, for any notional principal $\Np>0$ and weight $w\in [0,1]$, the {\it effective wealth} $\Vwe$ of a cross-currency portfolio satisfies $\Vwe_0=\Np=\alpha_0\Sd_0+\beta_0\Sfd_0$ and, for every $t\in [0,T]$,
\begin{equation*}
\Vwe_t=\alpha_0\Sd_t+\beta_0\Sfd_t=\Np(w\Sd_t(S^d_0)^{-1}+(1-w)\Sfd_t(\Sfd_0)^{-1})=\Np\Swe_t
\end{equation*}
where we denote
\begin{equation*}
\Swe_t:=\Np^{-1}\Vwe_t=w\Sd_t(S^d_0)^{-1}+(1-w)\Sfd_t(\Sfd_0)^{-1}=w(1+\Rd_t)+(1-w)(1+\Rfd_t)
\end{equation*}
so that, in particular, $\Swe_0=1$. We now consider the {\it basket process} $\Swe$ as a reference portfolio for an EPS
and thus also as the underlying process for related cross-currency call and put options.
For any notional principal $\Np>0$ and weight $w \in [0,1]$, the {\it effective rate of return} $\Rwe$ on a cross-currency portfolio equals, and $t\in [0,T]$,
\begin{equation*}
\Rwe_t:=\frac{\Vwe_t-\Vwe_0}{\Vwe_0}=\Swe_t-1= w\Rd_t+(1-w)\Rfd_t.
\end{equation*}

\begin{definition} \label{def4.1}
{\rm In a generic {\it cross-currency basket} EPS, the provider's cash flow at time $T$ is given in the domestic currency and equals
\begin{equation*}
\EPSwe:=\Np\ppa(\Rwe_T)=\Np\ppa\big(w\Rd_T+(1-w)\Rfd_T\big).
\end{equation*}
for some adjusted return function $\ppa :(-1,\infty ) \to\mathbb{R}$ with parameters $p_i,f_j,l_i$ and $g_j$.}
\end{definition}

To derive the static hedging strategy for a generic cross-currency basket EPS, we consider
cross-currency basket call and put options written on $\Swe$ with the payoffs $\whC_T(\Swe,K^d,T):=(\Swe_T-K^d)^+$ and $\whP_T(\Swe,K^d,T):=(K^d-\Swe_T)^+$ where the payoffs are given in the domestic currency.

\begin{corollary} \label{cor4.1}
The static hedging in terms of options written on $\Swe$ for a generic cross-currency basket EPS is given by Proposition \ref{pro2.1} with $\pp = \ppa$ and $X=\Swe $, that is, $\EPSwe=\whHwe$ where, for every $t\in [0,T]$,
\begin{equation*}
\whHwet=\Np\sum_{i=0}^{n}(p_{i+1}-p_i)\whP_t(\Swe,k^l_i,T)-\Np\sum_{j=0}^{m}(f_{j+1}-f_j)\whC_t(\Swe,k^g_j,T)
\end{equation*}
where $k^l_i=1+l_i$ and $k^g_j=1+g_j$ for every $i=0,1,\dots, n$ and $j=0,1,\dots ,m$.
\end{corollary}

Using this static hedging strategy, the provider can perfectly hedge the effective return on the whole reference portfolio, which includes both domestic equities and foreign equities struck in the domestic currency. Although the provider's hedging strategy is fairly simple, it hinges on a strong postulate that the relevant call and put options on the value process $\Swe$ are actively traded, which is a rather unrealistic postulate, in particular, due to the fact that $\Swe$ depends on the choice of the weight $w$. Therefore, we will also examine the theoretical pricing and hedging for a generic cross-currency EPS
within the framework of the market model introduced in Section \ref{sec3.2}, as well as its superhedging.

To this end, we start by noticing that
\begin{equation*}
\Swe_t=w\Sd_t(S^d_0)^{-1}+(1-w)\Sfd_t(\Sfd_0)^{-1}=w \whSd_t+(1-w)\whSfd_t
\end{equation*}
where, for computational convenience, we introduce normalized processes $\whSd_t:=\Sd_t(S^d_0)^{-1}$ and $\whSfd_t:=\Sfd_t(\Sfd_0)^{-1}$ so that $\whSd_0=\whSfd_0=1$. Of course, that normalization does not affect the dynamics of price processes $\Sd$ and $\Sfd$, which were specified in Section \ref{sec3.2}, but only their initial values and thus we have that
\begin{align*}
&d\whSd_t=\whSd_t\big( (r^d-\kappa^d)\,dt+\bm{\sigma}^d\,d\textbf{W}_t^d\big),\nonumber \\
&d\whSfd_t=\whSfd_t\big((r^d-\kappa^f)\,dt+(\bm{\sigma}^f+\bm{\sigma}^q)\,d\textbf{W}_t^d\big).
\end{align*}
The cross-currency basket call option written on $\Swe$ has the terminal payoff at maturity $T$ given by $\whC_T(\Swe,K^d,T)=(\Swe_T-K^d)^+$ in the domestic currency. Since $\Swe_T$ is a sum of lognormally distributed random variables it does not follow a lognormal distribution and thus the Black-Scholes-like pricing formula for a basket option is not available. Therefore, we analyze other ways of computing, at least approximately, the initial prices of cross-currency basket options within the model from Section \ref{sec3.2}:\\
(i) numerical pricing using the Monte Carlo method, \\
(i) approximate analytical pricing using the geometric averaging method,\\
(ii) approximate analytical pricing using the moment matching technique.

%%%%%%%%%%%%%%%%%%%%%%%%%%%%%%%%%%%%%%%%%%%%%%%%%%%%%%%%%%%%%%%%%%%%%%%%%%%%%%%%%%%%%%%%%%%%%%%%%%
\subsection{Approximate Pricing Through Geometric Averaging Method}    \label{sec4.2}
%%%%%%%%%%%%%%%%%%%%%%%%%%%%%%%%%%%%%%%%%%%%%%%%%%%%%%%%%%%%%%%%%%%%%%%%%%%%%%%%%%%%%%%%%%%%%%%%%%

For the analytical approximation of the price of the cross-currency basket call option, it is convenient to rewrite
the pricing equality as follows
\begin{align*}
\whC_0(\Swe,K^d,T)=e^{-r^d T}\,\EQd\big[(\Swe_T-K^d)^+\big]
=\EQd\big[\big(w\wtSd_T+(1-w)\wtSfd_T-\wt{K}^d\big)^+\big]
=\EQd\big[\big(\wtSwe_T-\wt{K}^d\big)^+\big]
\end{align*}
where we introduce the discounted values $\wtSd_t:=e^{-r^d t}\whSd_t$ and $\wtSfd_t:=e^{-r^d t}\whSfd_t$ and
$\wt{K}^d:=e^{-r^d T}K^d$ so that
\begin{align*}
&d\wtSd_t=\wtSd_t \big(-\kappa^d \, dt +\bm{\sigma}^d\,d\textbf{W}_t^d \big), \\
&d\wtSfd_t=\wtSfd_t\big(-\kappa^f \, dt+(\bm{\sigma}^f+\bm{\sigma}^q)\,d\textbf{W}_t^d \big),
\end{align*}
with $\wtSd_0=\wtSfd_0=1$.  To get an approximation for the price of a classical basket option on equities, Gentle \cite{G1993} proposed to replace the weighted arithmetic mean $\wtSwe_T:= w\wtSd_T +(1-w)\wtSfd_T$ by a similarly weighted geometric mean, which is henceforth denoted as $G_T$ and, in the present setup, is given by
\begin{align*}
G_T:=\big(\wtSd_T\big)^{w}\big(\wtSfd_T\big)^{(1-w)}
=e^{(w\bm{\sigma}^d+(1-w)(\bm{\sigma}^f+\bm{\sigma}^q))\textbf{W}_T^d-\frac{1}{2}\gamma T}=e^{\xi_T-\frac{1}{2}\gamma T}
\end{align*}
where
\begin{align*}
\gamma:= w \|\sigma^d\|^2+(1-w)\|\sigma^f+\sigma^q\|^2 + 2(w \kappa^d +(1-w) \kappa^f)T.
\end{align*}
Notice that $\xi_T$ is a Gaussian random variable under $\Qd$ with null expected value and the variance $v^2 T$ where
\begin{align*}
v^2:=\|w\bm{\sigma}^d+(1-w)(\bm{\sigma}^f+\bm{\sigma}^q)\|^2
=w^2\|\sigma^d\|^2+2w(1-w)\sigma^d(\sigma^f+\sigma^q)^*+(1-w)^2\|\sigma^f+\sigma^q\|^2.
\end{align*}
Consequently, the random variable $G_T$ has the lognormal distribution under $\Qd$ with the expected value
\begin{align*}
\EQd (G_T)=e^{\frac{1}{2}(v^2-\gamma )T}
%=\exp\Big(-\frac{1}{2}w(1-w)\big(\|\sigma^d\|^2-2\sigma^d(\sigma^f+\sigma^q)^*
%+\|\sigma^f+\sigma^q\|^2\big)T\Big)\\
=\exp \Big(-\frac{1}{2} w(1-w)\|\sigma^d-\sigma^f-\sigma^q\|^2T - (w \kappa^d +(1-w) \kappa^f)T\Big)=:\lambda.
\end{align*}
Then the proposed approximation for the initial price of the cross-currency basket call option is given by the expected value
\begin{equation*}
C^G_0(\wh{K}^d,T):=\EQd\big[\big(G_T-\wh{K}^d\big)^+\big]
\end{equation*}
where $\wh{K}^d$ equals
\begin{equation*}
\wh{K}^d:=\wt{K}^d+\EQd\big[G_T-\big(w\wtSd_T+(1-w)\wtSfd_T\big)\big]=\widetilde{K}^d+\lambda-\zeta
\end{equation*}
with
\begin{equation*}
\zeta :=\EQd\big[w\wtSd_T+(1-w)\wtSfd_T\big]=w\wtSd_0 e^{-\kappa^d T} +(1-w)\wtSfd_0 e^{-\kappa^f T}=
we^{-\kappa^d T} +(1-w) e^{-\kappa^f T},
\end{equation*}
which is a consequence of the dynamics of $\wtSd$ and $\wtSfd$ under $\Qd$. Elementary computations now lead to the
following result, which can be seen as a variant of a more general pricing formula for an option written on a basket on $n$ domestic equities in the multi-dimensional Black-Scholes model (see Proposition 6.9.1 in Musiela and Rutkowski \cite{MM2002}).

\begin{proposition} \label{pro4.1}
The model price $\whC_0(\Swe,K^d,T)$ of the cross-currency basket call option can be approximated by
\begin{equation*}
C^G_0(\wh{K}^d,T)=\lambda \Phin\big(\whd_+(\wh{K}^d,T)\big)-\wh{K}^d\Phin\big(\whd_-(\wh{K}^d,T)\big)
\end{equation*}
where
\begin{equation*}
\whd_{\pm}(\wh{K}^d,T)=\frac{1}{v\sqrt{T}}\bigg[\ln\frac{\lambda }{\wh{K}^d}\pm \frac{1}{2}v^2 T\bigg].
\end{equation*}
The model price $\whP_0(\Swe,K^d,T)$ of the cross-currency basket put option with the payoff
$\whP_T(\Swe,K^d,T)=(K^d-\Swe_T)^+$ can be approximated by
\begin{equation*}
P^G_0(\wh{K}^d,T)=\wh{K}^d\Phin\big(-\whd_-(\wh{K}^d,T)\big)-\lambda \Phin\big(-\whd_+(\wh{K}^d,T)\big).
\end{equation*}
\end{proposition}

\begin{remark} \label{rem4.1}
{\rm The put-call parity for cross-currency basket options has the form $\whC_0(\Swe,K^d,T)-\whP_0(\Swe,K^d,T)=1-\wt{K}^d$ (recall that $\wtSwe$ is a martingale under $\Q^d$ and $\wtSwe_0=1$) whereas for the approximated values it is given by $C^G_0(\wh{K}^d,T)-P^G_0(\wh{K}^d,T)=\lambda-\wh{K}^d$. Hence it suffices to focus on the pricing formula for a cross-currency basket call option.}
\end{remark}

\subsection{Approximate Pricing Through Moment Matching Method} \label{sec4.3}
%%%%%%%%%%%%%%%%%%%%%%%%%%%%%%%%%%%%%%%%%%%%%%%%%%%%%%%%%%%%%%%%%%%%%%%%%%%%%%%%%%%%%%%%%%%%%%%%%%

Borovkova et al. \cite{BPW2007} considered the problem of pricing basket options in a futures market and proposed a numerical approach based on matching three moments to approximate the distribution of baskets with possibly negative weights, which in turn yields a fairly accurate approximation for the price of a basket option. To be more specific, their method uses the first three moments to match the real-valued parameters $s,m,\tau,c$ of the approximating shifted lognormal distribution of the form
\begin{equation*}
c\big(e^{sZ+m}+\tau\big)
\end{equation*}
where $Z$ is a standard normal random variable under a probability measure $\Q$. The quantity $e^{m}$ can be interpreted as a scale parameter, the parameter $s$ represents a shape parameter, and $\tau$ is a location parameter. The parameter $c = \pm 1$ is included to adjust the overall sign of the basket price when weights are not all of the same sign and it is chosen to coincide with the sign of the skewness of the basket. For the pricing algorithm proposed in \cite{BPW2007}, a real-valued closed-form solution for the set of nonlinear equations was recently developed by Hu et al. \cite{HSV2023}
who obtained the closed-form approximation for the price of a basket option in terms of the mean, variance, and skewness
of the basket's distribution at maturity $T$. We are going to adapt their results to the present context.

Recall from Section \ref{sec4.1} that our goal is to compute
\begin{align*}
\whC_0(\Swe,K^d,T)=\EQd\big[\big(w\wtSd_T+(1-w)\wtSfd_T-\wt{K}^d\big)^+\big]=\EQd\big[\big(\wtSwe_T-\wt{K}^d\big)^+\big]
\end{align*}
where $\wt{K}^d:=e^{-r^d T}K^d$ and $\wtSwe_t:=w\wtSd_t+(1-w)\wtSfd_t$ for every $t\in [0,T]$. We have
\begin{align*}
&d\wtSd_t=\wtSd_t \big(-\kappa^d \, dt +\bm{\sigma}^d\,d\textbf{W}_t^d \big), \\
&d\wtSfd_t=\wtSfd_t\big(-\kappa^f \, dt+(\bm{\sigma}^f+\bm{\sigma}^q)\,d\textbf{W}_t^d \big),
\end{align*}
with $\wtSd_0=\wtSd_0=1$ and thus, in particular,
\begin{equation*}
\wtSwe_T=w e^{\bm{\sigma}^d\textbf{W}_T^d-\frac{1}{2}\| \bm{\sigma}^d \|^2T - \kappa^d T}
+(1-w)e^{(\bm{\sigma}^f+\bm{\sigma}^q)\textbf{W}_T^d-\frac{1}{2}\|\bm{\sigma}^f+\bm{\sigma}^q\|^2T-\kappa^f T}.
\end{equation*}

The first step hinges on computation of the first three moments of the random variable $\wtSwe_T$.
The proof of Lemma \ref{lem4.1} is elementary and thus it is omitted.

\begin{lemma} \label{lem4.1}
Let $w_1=w e^{-\kappa^d T},\,w_2=(1-w) e^{-\kappa^f T},\,\bm{\sigma}_1=\bm{\sigma}^d$ and $\bm{\sigma}_2=\bm{\sigma}^f+\bm{\sigma}^q$.
Then the first three moments of the random variable $\wtSwe_T$ under the martingale measure $\Q^d$ are
\begin{align*}
\wt{M}_1&=w_1+w_2, \\
\wt{M}_2&=w_1^2e^{\|\bm{\sigma}_1\|^2T}+2w_1w_2e^{\bm{\sigma}_1\bm{\sigma}_2^*T}
+w_2^2e^{\|\bm{\sigma}_2\|^2 T}, \\
\wt{M}_2&=w_1^3e^{3\|\bm{\sigma}_1\|^2 T}
+3w_1^2w_2e^{\big(\|\bm{\sigma}_1\|^2+2\bm{\sigma}_1\bm{\sigma}_2^*\big)T}+3w_1 w_2^2 e^{\big(\|\bm{\sigma}_2\|^2+2\bm{\sigma}_1\bm{\sigma}_2^*\big)T}+w_2^3 e^{3\|\bm{\sigma}_2\|^2T}.
\end{align*}
\end{lemma}

Let us denote the mean of $\wtSwe_T$ by $\mu_T:=\wt{M}_1$, the standard deviation by $\sigma_T^2 :=\wt{M}_2-\mu_T^2$,
and the skewness by
\begin{equation*}
\eta_T:=\frac{\wt{M}_3-3\mu_T\sigma_T^2-\mu_T^3}{\sigma_T^3}.
\end{equation*}
It can be checked that first three moments of the random variable $c (e^{sZ+m}+\tau)$ are
\begin{align*}
M_1(c,s,m,\tau)&:=c \big(e^{\frac{1}{2}s^2+m}+\tau\big), \nonumber \\
M_2(c,s,m,\tau)&:=c^2\big(e^{2s^2+2m}+2\tau e^{\frac{1}{2}s^2+m}+\tau^2\big), \\
M_3(c,s,m,\tau)&:=c^3\big(e^{\frac{9}{2}s^2+3m}+3\tau e^{2s^2+2m}+3\tau^2 e^{\frac{1}{2}s^2+m}+\tau^3\big).\nonumber
\end{align*}
Then the parameters $c,s,m,\tau,$ can be obtained by solving the moment matching conditions
$\wt{M}_i=M_i(c,s,m,\tau)$ for $i=1,2,3$. The following result was proven in Hu et al. \cite{HSV2023}.

\begin{lemma} \label{lem4.2}
A particular set of real-valued solution $c,s,m,\tau$ that satisfy the moment matching conditions $\wt{M}_i=M_i(c,s,m,\tau)$ for $i=1,2,3$ is given by the following functions of $(\mu_T,\sigma_T,\eta_T)$
\begin{equation*}
c=\sgn(\eta_T),\quad s=(\ln(x))^{\frac{1}{2}},\quad m=\frac{1}{2}\ln{\bigg(\frac{\sigma_T^2}{x(x-1)}\bigg)}, \quad \tau=\sgn(\eta_T)\mu_T-\frac{\sigma_T}{(x-1)^{\frac{1}{2}}},
\end{equation*}
where $\sgn$ is the sign function and
\begin{equation*}
x =\sqrt[3]{1+\frac{1}{2}\eta_T^2+\eta_T\sqrt{1+\frac{1}{4}\eta_T^2}}
   +\sqrt[3]{1+\frac{1}{2}\eta_T^2-\eta_T\sqrt{1+\frac{1}{4}\eta_T^2}} -1.
\end{equation*}
\end{lemma}

In view of Lemma \ref{lem4.2}, it is clear that  the parameters  $c,s,m,\tau$ are given by explicitly known functions of moments $\wt{M}_i,\, i=1,2,3$ and hence also by functions of market parameters $\kappa^d,\kappa^f,\bm{\sigma}^d,\bm{\sigma}^f,\bm{\sigma}^q$  and portfolio's weights $(w,1-w)$.
Furthermore, using the values of parameters $c,s,m,\tau$ from Lemma \ref{lem4.2}, the initial price of the basket call option can be approximated by
\begin{equation*}
\wh{C}(\wt{K}^d,T):=\EQ \big[\big(c(e^{sZ+m}+\tau) -\wt{K}^d\big)^+\big]
\end{equation*}
or, more explicitly, if $\eta_T>0$
\begin{equation*}
\wh{C}(\wt{K}^d,T):=\EQ \big[\big(e^{sZ+m}+\tau-\wt{K}^d\big)^+\big]
\end{equation*}
and, if $\eta_T<0$
\begin{equation*}
\wh{C}(\wt{K}^d,T):=\EQ \big[\big(-e^{sZ+m}-\tau-\wt{K}^d\big)^+\big]
\end{equation*}
where $Z$ is a standard normal random variable under the probability measure $\Q$.
Then the approximation for the price of the cross-currency basket call option is given by the following
result, which is an immediate consequence of Theorem 2.3 in Hu et al. \cite{HSV2023}.

\begin{lemma} \label{lem4.3}
The following equalities hold
\begin{equation*}
\wh{C}(\wt{K}^d,T)=
\begin{cases}
e^{m+\frac{1}{2}s^2}+\tau-\wt{K}^d,\quad &c=1,\,\wt{K}^d \leq \tau, \\
e^{m+\frac{1}{2}s^2}\Phin(d_{11})-(\wt{K}^d-\tau)\Phin(d_{12}),\quad &c=1,\,\wt{K}^d>\tau, \\
0, \quad &c=-1,\,\wt{K}^d>-\tau, \\
-e^{m+\frac{1}{2}s^2}\Phin(d_{21})+(-\wt{K}^d-\tau)\Phin(d_{22}),\quad &c=-1,\,\wt{K}^d\leq-\tau,
\end{cases}
\end{equation*}
where the parameters $c,s,m,\tau$ are given by Lemma \ref{lem4.2} and
\begin{align*}
d_{11}&=\frac{-\ln(\wt{K}^d-\tau)+m+s^2}{s},\quad d_{12}=\frac{-\ln(\wt{K}^d-\tau)+m}{s}, \\
d_{21}&=\frac{\ln(-\wt{K}^d-\tau)-m-s^2}{s},\quad d_{22}=\frac{\ln(-\wt{K}^d-\tau)-m}{s}.
\end{align*}
\end{lemma}

To approximate the price of the cross-currency basket put option, one can employ the classical put-call parity
and hence obtain also an approximation for the price of a cross-currency EPS.

%%%%%%%%%%%%%%%%%%%%%%%%%%%%%%%%%%%%%%%%%%%%%%%%%%%%%%%%%%%%%%%%%%%%%%%%%%%%%%%%%%%%%%%%
\subsection{Superhedging Cost of the Buffer Cross-Currency Basket EPS} \label{sec4.4}
%%%%%%%%%%%%%%%%%%%%%%%%%%%%%%%%%%%%%%%%%%%%%%%%%%%%%%%%%%%%%%%%%%%%%%%%%%%%%%%%%%%%%%%%

It should be acknowledged there are still some weaknesses in the concept of the static hedging strategy for a
cross-currency EPS since it may be not easily available in practice. Market makers are more likely to use either a particular equity or a market index as an underlying asset for liquidly traded options, whereas bespoke options with the holder's portfolio as the underlying asset can only be offered on OTC market. Although Corollary \ref{cor4.1} shows that, in theory, a perfect static hedging of a generic EPS on a cross-currency reference portfolio is feasible, it may be difficult to implement the static hedging even for most basic cross-currency EPSs. For that reason, we will also study some superhedging strategies for cross-currency EPSs. Obviously, another viable option for the provider is to choose a stochastic model for the cross-currency market and use it to implement a dynamic hedging strategy based on actively traded assets and futures contracts on currencies.

\begin{lemma} \label{lem4.4}
The following inequalities hold, for every $(x,y)\in \bbR^2$,
\begin{align*}
&(x+y)^+\leq x^+ + y^+ , \\
&(x+y)^+\geq (x+y)\I_{\{x\geq 0,\,y\geq 0\}}=x^+\I_{\{y\geq 0\}}+y^+\I_{\{x\geq 0\}},
\end{align*}
where the first inequality is strict on the set $A:=\{x<0,\,y>0\}\cup\{x>0,\,y<0\}$ and the second one is strict
on the set $B:=\{x<0,\,x+y>0\}\cup\{y <0,\,x+y >0\}.$
\end{lemma}

Our goal is to construct a particular superhedging strategy for the cross-currency buffer EPS by using options with the underlying assets $\Sd$ and $\Sfd$ introduced in Sections \ref{sec3.3} and \ref{sec3.4}. Recall that the aggregate effective return satisfies $\Swe_T=w \whSd_T+(1-w)\whSfd_T$. Hence Lemma \ref{lem4.4} leads to the following result.

\begin{lemma} \label{lem4.5}
The following inequalities hold, for every $w \in [0,1]$,
\begin{align*}  % \label{ineq1}
&\big(k^l_{i}-\Swe_T\big)^+ \leq w\big(k^l_{i}-\whSd_T\big)^+ + (1-w)\big(k^l_{i}-\whSfd_T\big)^+, \nonumber \\
&\big(k^l_{i}-\Swe_T\big)^+ \geq w\big(k^l_{i}-\whSd_T\big)^+ \I_{\{\whSfd_T\leq k^l_{i}\}}
   +(1-w)\big(k^l_{i}-\whSfd_T\big)^+\I_{\{\whSd_T\leq k^l_{i}\}}, \nonumber \\
&\big(\Swe_T- k^g_{j}\big)^+ \leq w\big(\whSd_T-k^g_{j}\big)^+ + (1-w)\big(\whSfd_T-k^g_{j}\big)^+, \nonumber \\
&\big(\Swe_T- k^g_{j}\big)^+ \geq w\big(\whSd_T-k^g_{j}\big)^+ \I_{\{\whSfd_T\geq k^g_{j}\}}
   +(1-w)\big(\whSfd_T-k^g_{j}\big)^+\I_{\{\whSd_T\geq k^g_{j}\}}. \nonumber
\end{align*}
\end{lemma}

\begin{proof}
For the first two inequalities, we observe that
\begin{equation*}
k^l_{i}-\Swe_T=w(k^l_{i}-\whSd_T)+(1-w)(k^l_{i}-\whSfd_T)
\end{equation*}
and we apply Lemma \ref{lem4.4}. For the last two, we use the equality
\begin{equation*}
\Swe_T- k^g_{j}= w(\whSd_T- k^g_{j})+(1-w)(\whSfd_T-k^g_{j})
\end{equation*}
and, once again, we make use of Lemma \ref{lem4.4}.
\end{proof}

It is worth noting that the first inequality in Lemma \ref{lem4.5} is strict on the event
\begin{equation*}
A^l_{i}:=\{\whSd_T> k^l_{i},\,\whSfd_T<k^l_{i}\}\cup\{\Sd_T<k^l_{i},\,\whSfd_T>k^l_{i}\}
\end{equation*}
and the second one is strict on
\begin{equation*}
B^l_{i}:=\{\whSd_T>k^l_{i},\,\Swe_T>k^l_{i}\}\cup\{\whSfd_T>k^l_{i},\,\Swe_T>k^l_{i}\}.
\end{equation*}
The third one is strict on the event
\begin{equation*}
A^g_{j}:=\{\whSd_T<k^g_{j},\,\whSfd_T>k^g_{j}\}\cup\{\whSd_T> k^g_{j},\,\whSfd_T<k^g_{j}\}.
\end{equation*}
and the last one is strict on
\begin{equation*}
B^g_{j}:=\{\whSd_T<k^g_{j},\,\Swe_T>k^g_{j}\}\cup\{\whSfd_T<k^g_{j},\,\Swe_T>k^g_{j}\}.
\end{equation*}

Lemma \ref{lem4.5} can be used to construct a particular superhedging strategy for the buffer cross-currency basket EPS.
From Corollary \ref{cor4.1}, we obtain
\begin{equation*}
\EPSwe =\Np \big(p_2\whP_T(\Swe,k^l_1,T)-f_2\whC_T(\Swe,k^g_1,T)\big)
\end{equation*}
where $k^l_1=1+l_1$ and $k^g_1=1+g_1$. Hence it suffices to set $\Np=1$ and consider the contingent claim
\begin{align*}
Y:=\pp_B (\Swe_T)=p_2\big(k^l_1-\Swe_T\big)^+ - f_2\big(\Swe_T- k^g_1\big)^+
\end{align*}
where $p_2$ and $f_2$ are strictly positive numbers. Our goal is to find a contingent claim $\whY$ such that $\whY \geq Y$ and static hedging of $\whY$ is feasible with relatively simple options. By combining the first inequality with the last in Lemma \ref{lem4.5}, we obtain
\begin{align*}
Y&=p_2\big(k^l_1-\Swe_T\big)^+ - f_2\big(\Swe_T- k^g_1\big)^+ \leq  p_2 w\big(k^l_1-\whSd_T\big)^+ - f_2 w\big(\whSd_T-k^g_1\big)^+\I_{\{\whSfd_T>k^g_1\}}  \\
&+p_2(1-w)\big(k^l_1-\whSfd_T\big)^+ -f_2(1-w)\big(\whSfd_T-k^g_1\big)^+\I_{\{\whSd_T>k^g_1\}}=:\whY.
\end{align*}
Since $\whY$ dominates $Y$ it is clear that a replicating strategy for $\whY$ is also a superhedging strategy for $Y$
and, in fact, it is a strict superhedging of $Y$ on the event $A^l_1 \cup B^g_1$. Let us denote such a strategy
by $\Hwesb$ so that $\HwesbT=\whY$. Since our goal is to obtain a static trading strategy, we introduce
two instances of cross-currency European call options with the payoff in the domestic currency given by
\begin{equation*}
C^*_T(\whSd,K,T,\whSfd):=\big(\whSd_T-K\big)^+\I_{\{\whSfd_T\geq K\}}
\end{equation*}
and
\begin{equation*}
C^*_T(\whSfd,K,T,\whSd):=\big(\whSfd_T-K\big)^+\I_{\{\whSd_T\geq K\}}.
\end{equation*}
The following corollary to Lemma \ref{lem4.5} gives a superhedging cost for the buffer cross-currency basket EPS

\begin{corollary} \label{cor4.2}
Let $\Hwesb$ be a static superhedging hedging strategy for the buffer cross-currency basket EPS such that $\HwesbT=\whY$.
Then the initial cost of $\Hwesb$ equals
\begin{align*}
\Hwesb_0&= w \Big( p_2 P_0(\whSd,k^l_1,T) - f_2 C^*_0(\whSd,k^g_1,T,\whSfd) \Big) \\
 &+(1-w) \Big( p_2\wtP_0(\whSfd,k^l_1,T) - f_2 C^*_0(\whSfd,k^g_1,T,\whSd) \Big).
\end{align*}
\end{corollary}

To find the closed-form expression for the superhedging cost $\Hwesb_0$ within the cross-currency market model introduced in Section \ref{sec3.2}, it suffices to compute the initial model prices $C^*_0(\whSd,k^g_1,T,\whSfd)$ and $C^*_0(\whSfd,k^g_1,T,\whSd)$ of cross-currency call options. Let $\Phin_2(x_1,x_2,\rho)$ denote the c.d.f. of the two-dimensional standard Gaussian distribution with the correlation coefficient $\rho$.
We set  % $\sigma_1:=\|\bm{\sigma}^d\|,\,\sigma_2:=\|\bm{\sigma}^f+\bm{\sigma}^q\|$ and
\begin{align*}
\rho:=\frac{(\bm{\sigma}^f+\bm{\sigma}^q)(\bm{\sigma}^d)^*}{\|\bm{\sigma}^d\|\|\bm{\sigma}^f+\bm{\sigma}^q\|}.
\end{align*}

\begin{proposition} \label{pro4.2}
The following equality holds, for every $K>0$,
\begin{align*}
C^*_0(\whSd,K,T,\whSfd)=\Phin_2(g_1,g_2,\rho)-\wtK \Phin_2(d_1,d_2,\rho)
\end{align*}
where $\wtK = e^{-r^d T}K$ and
\begin{align*}
d_1&=\frac{-\ln\wtK-\kappa^d T-\frac{1}{2}\|\bm{\sigma}^d\|^2T}{\|\bm{\sigma}^d\|\sqrt{T}},\quad d_2=\frac{-\ln\wtK-\kappa^f T-\frac{1}{2}\|\bm{\sigma}^f+\bm{\sigma}^q\|^2T}{\|\bm{\sigma}^f+\bm{\sigma}^q\|\sqrt{T}}, \\
g_1&=\frac{-\ln\wtK-\kappa^d T +\frac{1}{2}\|\bm{\sigma}^d\|^2T}{\|\bm{\sigma}^d\|\sqrt{T}}, \quad
g_2=\frac{-\ln\wtK-\kappa^f T+(\bm{\sigma}^f+\bm{\sigma}^q)(\bm{\sigma}^d)^*T     -\frac{1}{2}\|\bm{\sigma}^f+\bm{\sigma}^q\|^2T}{\|\bm{\sigma}^f+\bm{\sigma}^q\|\sqrt{T}}.
\end{align*}
\end{proposition}

\begin{proof}
Recall that $\whSd$ and $\whSfd$ satisfy under the domestic martingale measure $\Qd$
\begin{align*}
&d\whSd_t=\whSd_t\big( (r^d-\kappa^d)\,dt+\bm{\sigma}^d\,d\textbf{W}_t^d\big),\nonumber \\
&d\whSfd_t=\whSfd_t\big((r^d-\kappa^f)\,dt+(\bm{\sigma}^f+\bm{\sigma}^q)\,d\textbf{W}_t^d\big).
\end{align*}
with initial conditions $\whSd_0=\whSfd_0=1$. We observe that
\begin{align*}
&C^*_0(\whSd,K,T,\whSfd)
=e^{-r^d T}\,\EQd\big[\big(\whSd_T-K\big)^+\I_{\{\whSfd_T\geq K\}}\big]
=\EQd\big[\big(\wtSd_T-\wtK\big)^+\I_{\{\whSfd_T\geq K\}}\big]\\
&=\EQd\big[\big(\wtSd_T-\wtK\big)\I_{\{\whSd_T>K,\,\whSfd_T>K\}}\big]
=\EQd\big[\wtSd_T\I_{\{\whSd_T>K,\,\whSfd_T>K\}}\big]
 - \wtK \Qd \big[\whSd_T>K,\,\whSfd_T>K\big] =: I_1 - \wtK I_2
\end{align*}
where $\wtSd_t:=e^{-r^d t}\whSd_t$ so that $d\wtSd_t=-\wtSd_t\kappa^d\,dt+\wtSd_t\bm{\sigma}^d\,d\textbf{W}_t^d$ and $\wtK:=e^{-r^d T}K$. Let us introduce the probability measure $\wtQd$ equivalent to $\Qd$ on $(\Omega,\cF_T)$ by the Radon-Nikodym density
\begin{align*}
\frac{d\wtQd}{d\Qd}:=e^{\bm{\sigma}^d\,\textbf{W}_T^d -\frac{1}{2}\|\bm{\sigma}^d\|^2 T}.
\end{align*}
Then the process $\wt{\textbf{W}}_t^d :=\textbf{W}_t^d -(\bm{\sigma}^d)^* t,\, t \in [0,T]$ is a standard Brownian motion under $\wtQd$ and the dynamics of $\whSd$ and $\whSfd$ under $\wtQd$ are
\begin{align*}
&d\whSd_t=\whSd_t\big( (r^d-\kappa^d+\|\bm{\sigma}^d\|^2)\,dt+\bm{\sigma}^d \,d\wt{\textbf{W}}_t^d \big),\\
&d\whSfd_t=\whSfd_t\big(\big(r^d-\kappa^f+(\bm{\sigma}^f+\bm{\sigma}^q)(\bm{\sigma}^d)^*\big)\,dt+(\bm{\sigma}^f+\bm{\sigma}^q)\,d\wt{\textbf{W}}_t^d\big).
\end{align*}
It is now easy to check that
\begin{align*}
&I_1=\wtQd \big[\whSd_T>K,\,\whSfd_T>K\big]=\Phin_2(g_1,g_2,\rho),\\
&I_2=\Qd \big[\whSd_T>K,\,\whSfd_T>K\big]=\Phin_2(d_1,d_2,\rho),
\end{align*}
with the parameters $g_1,g_2,d_1$ and $d_2$ are given in the statement of the lemma.
\end{proof}

The result for $C^*_0(\whSfd,K,T,\whSd)$ is analogous Proposition \ref{pro4.2}, although it requires minor adjustments to the notation.

\begin{proposition} \label{pro4.3}
The following equality holds, for every $K>0$,
\begin{align*}
C^*_0(\whSfd,K,T,\whSd)=\Phin_2(h_1,h_2,\rho)-\wtK\Phin_2(d_1,d_2,\rho)
\end{align*}
where $d_1$ and $d_2$ are given in Proposition \ref{pro4.2} and
\begin{align*}
h_1=\frac{-\ln\wtK-\kappa^d T+(\bm{\sigma}^f+\bm{\sigma}^q)(\bm{\sigma}^d)^*T-\frac{1}{2}\|\bm{\sigma}^d\|^2T}{\|\bm{\sigma}^d\|\sqrt{T}},\quad
h_2=\frac{-\ln\wtK-\kappa^f T+\frac{1}{2}\|\bm{\sigma}^f+\bm{\sigma}^q\|^2T}{\|\bm{\sigma}^f+\bm{\sigma}^q\|\sqrt{T}}.
\end{align*}
\end{proposition}

\begin{proof}
We have
\begin{align*}
&C^*_0(\whSfd,K,T,\whSd)=e^{-r^d T}\,\EQd\big[\big(\whSfd_T-K\big)^+\I_{\{\whSd_T\geq K\}}\big]
=\EQd\big[\big(\wtSfd_T-\wtK\big)^+\I_{\{\whSd_T\geq K\}}\big]\\
&=\EQd\big[\big(\wtSfd_T-\wtK\big)\I_{\{\whSd_T>K,\,\whSfd_T>K\}}\big]
=\EQd\big[\wtSfd_T\I_{\{\whSd_T>K,\,\whSfd_T>K\}}\big]-\wtK\Qd\big[\whSd_T>K,\,\whSfd_T>K\big] =: \wh{I}_1 - \wtK I_2.
\end{align*}
We introduce an equivalent probability measure $\whQd$ on $(\Omega,\cF_T)$ by the Radon-Nikodym density
\begin{align*}
\frac{d\whQd}{d\Qd}:=e^{(\bm{\sigma}^f+\bm{\sigma}^q)\,\textbf{W}_T^d -\frac{1}{2}\|\bm{\sigma}^f+\bm{\sigma}^q\|^2 T}.
\end{align*}
Then the process $\wh{\textbf{W}}_t^d :=\textbf{W}_t^d -(\bm{\sigma}^f+\bm{\sigma}^q)^* t,\, t \in [0,T]$ is a standard Brownian motion under $\whQd$ and the dynamics of $\whSd$ and $\whSfd$ under $\whQd$ are
\begin{align*}
&d\whSd_t=\whSd_t\big( (r^d-\kappa^d) \, dt + (\bm{\sigma}^f+\bm{\sigma}^q)(\bm{\sigma}^d)^*\,dt+\bm{\sigma}^d \,d\wh{\textbf{W}}_t^d \big),\\
&d\whSfd_t=\whSfd_t\big( (r^d-\kappa^f+\|\bm{\sigma}^f+\bm{\sigma}^q\|^2)\,dt+(\bm{\sigma}^f+\bm{\sigma}^q)\,d\wh{\textbf{W}}_t^d\big),
\end{align*}
and thus the stated expressions are valid.
\end{proof}

%%%%%%%%%%%%%%%%%%%%%%%%%%%%%%%%%%%%%%%%%%%%%%%%%%%%%%%%%%%%%%%%%%%%%%%%%%%%%%%%%%%%%%%%
\subsection{Superhedging Cost of the Floor Cross-Currency Basket EPS} \label{sec4.5}
%%%%%%%%%%%%%%%%%%%%%%%%%%%%%%%%%%%%%%%%%%%%%%%%%%%%%%%%%%%%%%%%%%%%%%%%%%%%%%%%%%%%%%%%

Our next goal is to study a static superhedging strategy for the floor cross-currency basket EPS. Setting $\Np=1$, we now have that
\begin{equation*}
\EPSwe = \psi_{F}(\Swe_T)= -p_1\,\whP_T(\Swe,k^l_1,T)+p_1\whP_T(\Swe,1,T)-f_2\,\whC_t(\Swe,k^g_1,T)
\end{equation*}
and we denote
\begin{equation*}
P^*_T(\whSd,K,T,\whSfd):=\big(K-\whSd_T\big)^+\I_{\{\whSfd_T \leq K\}}
\end{equation*}
and
\begin{equation*}
P^*_T(\whSfd,K,T,\whSd):=\big(K-\whSfd_T\big)^+\I_{\{\whSd_T\leq K\}}.
\end{equation*}
The following corollary to Lemma \ref{lem4.5} gives the initial cost of a static superhedging strategy for
the cross-currency floor EPS.

\begin{corollary} \label{cor4.3}
The initial cost of a static superhedging strategy $\Hwesf$ for the floor cross-currency basket EPS with $\Np=1$ satisfies
\begin{align*}
\Hwesf_0&= w \Big( p_1P_0(\whSd,1,T)-p_1 P^*_0(\whSd,k^l_1,T)-f_2C^*_0(\whSd,k^g_1,T,\whSfd) \Big) \\
 &+(1-w) \Big( p_1\wtP_0(\whSfd,1,T)-f_2 C^*_0(\whSfd,k^g_1,T,\whSd)
 -p_1 P^*_0(\whSfd,k^l_1,T,\whSd) \Big).
\end{align*}
\end{corollary}

In view of Corollary \ref{cor4.3}, in order to find an explicit representation for the superhedging cost $\Hwesf_0$, it suffices to combine Propositions \ref{pro4.2} and \ref{pro4.3} with the following result.  % $h_1$ and $h_2$ are given in Proposition \ref{pro4.3}.

\begin{proposition} \label{pro4.4}
The following equalities are valid, for every $K>0$,
\begin{align*}
&P^*_0(\whSd,K,T,\whSfd)=\wtK\Phin_2(-d_1,-d_2,\rho)-\Phin_2(-g_1,-g_2,\rho), \\
&P^*_0(\whSfd,K,T,\whSd)=\wtK\Phin_2(-d_1,-d_2,\rho)-\Phin_2(-h_1,-h_2,\rho).
\end{align*}
\end{proposition}

\begin{proof}
We first compute
\begin{align*}
&P^*_0(\whSd,K,T,\whSfd)
=e^{-r^d T}\,\EQd\big[\big(K-\whSd_T\big)^+\I_{\{\whSfd_T\leq K\}}\big]=\EQd\big[\big(\wtK-\wtSd_T\big)^+\I_{\{\whSfd_T\leq K\}}\big]\\
%&=\EQd\big[\big(\wtSd_T-\wtK\big)\I_{\{\whSd_T<K,\,\whSfd_T<K\}}\big]
&=\wtK \Qd \big[\whSd_T<K,\,\whSfd_T<K\big]-\EQd\big[\wtSd_T\I_{\{\whSd_T<K,\,\whSfd_T<K\}}\big].
\end{align*}
To obtain the first equality we use the change of a probability measure from the proof of Proposition \ref{pro4.2}.
For the second equality, we observe that
\begin{align*}
&P^*_0(\whSfd,K,T,\whSd)=e^{-r^d T}\,\EQd\big[\big(K-\whSfd_T\big)^+\I_{\{\whSd_T\leq K\}}\big]\\
%=\EQd\big[\big(\wtK-\wtSfd_T\big)^+\I_{\{\whSd_T\leq K\}}\big]\\
%&=\EQd\big[\big(\wtK-\wtSfd_T\big)\I_{\{\whSd_T<K,\,\whSfd_T<K\}}\big]
&=\wtK\Qd\big[\whSd_T<K,\,\whSfd_T<K\big]-\EQd\big[\wtSfd_T\I_{\{\whSd_T<K,\,\whSfd_T<K\}}\big]. %=:\wtK J_2-\wh{J}_1.
\end{align*}
and we use the change of a probability measure introduced in the proof of Proposition \ref{pro4.3}.
\end{proof}

%%%%%%%%%%%%%%%%%%%%%%%%%%%%%%%%%%%%%%%%%%%%%%%%%%%%%%%%%%%%%%%%%%%%%%%%%%%%%%%%%%%%%%%%%%%%
\subsection{Approximate Hedging of the Buffer Cross-Currency Basket EPS} \label{sec4.6}
%%%%%%%%%%%%%%%%%%%%%%%%%%%%%%%%%%%%%%%%%%%%%%%%%%%%%%%%%%%%%%%%%%%%%%%%%%%%%%%%%%%%%%%%%%%%

The superhedging strategy $\Hwes$ combined with the premium $\wh{c}$ ensures that the provider will never lose money
and thus it can be seen as a good choice for the provider. However, the value $\wh{c}$ is not necessarily competitive and may be too high for the buyer of the EPS since it does not represent an arbitrage-free price. It thus make sense to search for a less expensive hedging strategy even if it may generate a loss in some scenarios. For lack of a better term, we will refer to that situation as an {\it approximate hedging}. Of course, we would like to find an approximate hedging with some desirable properties, such as closeness to superhedging and convenience of implementation in practice.

If the indicator functions in definition of $\whY$ are omitted, then we obtain the claim
\begin{align*}
\wtY=w[p_2\big(k^l_1-\Sd_T\big)^+ - f_2\big(k^g_1-\Sd_T\big)^+]+(1-w)[p_2\big(k^l_1-\Sfd_T\big)^+ -f_2\big(k^g_1-\Sfd_T\big)^+]
\end{align*}
which coincides with $\whY$ on the complement of the event $B^g_1$. On the one hand, it should be acknowledged that the claim $\wtY$ does not dominate $Y$, in general. On the other hand, to replicate $\wtY$ it suffices to use European options with the underlying assets $\Sd$ and $\Sfd$ to build the hedging portfolios from Propositions \ref{cor3.1} and \ref{cor3.3}. Then
an approximate replication is obtained by combining the domestic and effective foreign static hedging strategies with weights $w$ and $1-w$, respectively. Formally, the wealth $\Hwea$ of the approximate hedging strategy for
the cross-currency buffer EPS satisfies
\begin{equation*}
\Hwea_T=w\whHd+(1-w)\whHfe
\end{equation*}
where $\whHd$ and $\whHfe$ are given by Corollaries \ref{cor3.1} and \ref{cor3.3}, respectively. The strategy $\Hwea$ does
not satisfy the superhedging condition $\Hwea_T \geq \whHwe$ and thus it may result in provider's losses in some adverse market scenarios. However, it only requires to use on simple options on reference portfolios $\Sd$ and $\Sf$.

\begin{corollary} \label{cor4.4}
The initial cost of the approximate hedging strategy $\Hwea$ for the buffer cross-currency basket EPS with $\Np=1$ equals
\begin{equation*}
\Hwea_0=w\whH^d_0+(1-w)\whH^{f,e}_0.
\end{equation*}
\end{corollary}

Notice that the static hedging strategy from Corollary \ref{cor4.2} has a similar structure to two separate hedging strategies from Propositions \ref{cor3.1} and \ref{cor3.3}. The main differences between the strategy from Corollary \ref{cor4.2}
and separate strategies are the underlying assets and the strikes of European options used in hedging portfolios. In separate hedging strategies, the provider considers domestic equity and foreign equity independently and thus the strikes of options are only related to each economy and, as a consequence, this also applies to the combined superhedging strategy. In contrast, a static hedging strategy given by $\whHwe$ hinges on the cross-currency reference portfolio as an underlying asset and thus strikes are based on the whole portfolio.

%%%%%%%%%%%%%%%%%%%%%%%%%%%%%%%%%%%%%%%%%%%%%%%%%%%%%%%%%%%%%%%%%%%%%%%%%%%%%%%%%%%%
%%%%%%%%%%%%%%%%%%%%%%%%%%%%%%%%%%%%%%%%%%%%%%%%%%%%%%%%%%%%%%%%%%%%%%%%%%%%%%%%%%%%
\section{Numerical Studies and Backtesting} \label{sec5}
%%%%%%%%%%%%%%%%%%%%%%%%%%%%%%%%%%%%%%%%%%%%%%%%%%%%%%%%%%%%%%%%%%%%%%%%%%%%%%%%%%%%
%%%%%%%%%%%%%%%%%%%%%%%%%%%%%%%%%%%%%%%%%%%%%%%%%%%%%%%%%%%%%%%%%%%%%%%%%%%%%%%%%%%%

In numerical studies based on the cross-currency market model given by \eqref{eq4}, we use the model parameters from Xu et al. \cite{XLR2024}. In particular, the volatility of the domestic portfolio is set to be $\|\bm{\sigma}^d \|=10\%$, which is consistent with the level of implied volatility represented by the ASX~200 VIX index in 2024, and for the foreign portfolio the volatility is chosen to be $\|\bm{\sigma}^f\|=15\%$, which is the medium value of the implied volatility of the U.S. equity options market during the same period, as measured by S\&P~500 VIX index. The dividend yields are taken to be $\kappa^d=4\%$ and $\kappa^f=2\%$ and thus they are in accord with the dividend yields for stock indices ASX~200 and S\&P~500, respectively.  The exchange rate USD/AUD is set to be 1.58 (i.e., USD 1.00 costs AUD 1.58), which represent the approximate mean value of the exchange rate in April 2025, and the volatility of the exchange rate is assumed to be $\|\bm{\sigma}^q\|=9\%$, which is based on the recent data for USD/AUD currency options.

Regarding the choice of correlations, it's important to note that the relationship between the Australian and U.S. equity markets is highly sensitive to the time period and the return interval chosen. As reported in \cite{L2018}, from 2000 to 2018, the correlation between the daily returns on ASX~200 and S\&P~500 was just 0.14. However, our numerical study analyzing historical data from 2014 to 2024 found that the correlation between the one-year trailing returns on ASX~200 and S\&P~500 was 0.77. Therefore, we consider several values for the correlation $\rhodf$ between domestic and foreign equities: -0.4 (negative correlation), 0.1 (low correlation), and 0.7 (high correlation), as these provide a useful range for analysis.
The correlation between the exchange rate and equity market is also a notable phenomenon, although that relationship can be complex and thus its analysis requires a strong expertise in FX market. For concreteness, we set the correlation between the domestic equity $\Sd$ (resp., the foreign equity $\Sf$) and the exchange rate to be $\rhod=0.1$ (resp., $\rhof=-0.3$). This means, for example, that a falling AUD (i.e., an increasing exchange rate $\X$) is likely to be linked with in a rising ASX~200 quote (resp., a falling S\&P 500 quote), and vice versa.

To obtain the vectors of volatilities in dynamics \eqref{eq4} we introduce the correlated Brownian motions $Z^1,Z^2$ and $Z^3$ by
\begin{equation} \label{eq_Wiener}
Z^1_t= W^{d1}_t,\ \,
Z^2_t=\rhodf W^{d1}_t+\sqrt{1-\rhodf^2}\,W^{d2}_t,\ \,
Z^3_t=\alone W^{d1}_t+\altwo W^{d2}_t+\althree W^{d3}_t
\end{equation}
where $\bm{W}^d=(W^{d1},W^{d2},W^{d3})$ is a standard Brownian motion  under $\Q^d$ and where we denote
\begin{align} \label{eq_alphas}
\alone := \rhod,\quad \altwo:=\frac{\rhof-\rhodf\rhod}{\sqrt{1-\rhodf^2}},\quad \althree:= \sqrt{1-\rhod^2-\altwo^2}.
\end{align}
Then the parameters used in what follows are summarised in Table \ref{table1} and they are consistent with the set of volatilities and correlations chosen above.

%%%%%%%%%%%%%%%%%%%%%%%%%%%%%%%%%%%%%%%%%%%%%%%%%%%%%%%%%%%%%%%%%%%%%%%%%%%%%%%%%%%%%%%%%%%%%
\subsection{Model Pricing of Domestic and Foreign EPSs}   \label{sec5.1}
%%%%%%%%%%%%%%%%%%%%%%%%%%%%%%%%%%%%%%%%%%%%%%%%%%%%%%%%%%%%%%%%%%%%%%%%%%%%%%%%%%%%%%%%%%%%%

\begin{table} [h!]
\centering
\begin{tabular}{|l|c|}
\hline
Parameter & Value\\
\hline
Domestic short-term interest rate $r^d$ & 4.1\% \\
Foreign short-term interest rate $r^f$ & 4.5\% \\
Dividend yields of the domestic asset $\kappa^d$ & $4\%$ \\
Dividend yields of the foreign asset $\kappa^f$ & $2\%$ \\
Initial exchange rate $\X_0$ & 1.58 \\
Volatility of the domestic asset $\bm{\sigma}^d$ & $\|\bm{\sigma}^d\|=10\%$ \\
Volatility of the foreign asset $\bm{\sigma}^f$ & $\|\bm{\sigma}^f\|=15\%$ \\
Volatility of the exchange rate $\bm{\sigma}^q$ & $\|\bm{\sigma}^q\|=9\%$ \\
Correlation between two assets $\rhodf$ & $-0.4, 0.1, 0.7$ \\
Correlation between the domestic asset and exchange rate $\rhod$ & $0.1$ \\
Correlation between the foreign asset and exchange rate $\rhof$ & $-0.3$ \\
\hline
\end{tabular}
\caption{Parameters for cross-currency market model}
\label{table1}
\end{table}

% \begin{table}[h]
% \centering
% \begin{tabular}{|l|l|c|c|c|}
% \hline
% \multicolumn{2}{|c|}{\textbf{Category}} & \textbf{ASX~200 Index} & \textbf{S\&P~500 Index} & \textbf{Exchange Rate} \\
% \hline
% \multirow{5}{*}{\shortstack[l]{Expected \\ return \\ below}}
% & -5\%    & -0.0981 & -0.1092 & -0.0938 \\
% & -2.5\%  & -0.0820 & -0.0947 & -0.0733\\
% & 0\%     & -0.0673 & -0.0743 & -0.0523 \\
% & 2.5\%   & -0.0522 & -0.0567 & -0.0381 \\
% & 5\%     & -0.0308 & -0.0379 & -0.0239 \\
% \hline
% \multirow{5}{*}{\shortstack[l]{Median \\ return \\ below}}
% & -5\%    & -0.0900 & -0.1004 & -0.0750 \\
% & -2.5\%  & -0.0731 & -0.0864 & -0.0599 \\
% & 0\%     & -0.0605 & -0.0665 & -0.0393 \\
% & 2.5\%   & -0.0458 & -0.0433 & -0.0259 \\
% & 5\%     & -0.0222 & -0.0204 & -0.0149 \\
% \hline
% \multicolumn{2}{|l|}{Range}     & $[-0.2629, 0.4854]$ & $[-0.2107, 0.7612]$ & $[-0.2585, 0.3544]$ \\
% \hline
% \multicolumn{2}{|l|}{Mean}  & 0.0398 & 0.1106 & 0.0332 \\
% \hline
% \multicolumn{2}{|l|}{Median}  & 0.0425 & 0.1212 & 0.0319 \\
% \hline
% \multicolumn{2}{|l|}{Variance}  & 0.1002 & 0.1341 & 0.0884 \\
% \hline
% \end{tabular}
% \caption{Descriptive statistics of one-year trailing returns on both indices and exchange rate in 2014--2024}
% \label{quan_data}
% \end{table}

After defining the parameters of domestic and foreign markets, we first examine the pricing of some standard EPSs on separate domestic and foreign portfolios. We focus here on two particular instances of a standard EPS, the buffer EPS and the floor EPS, and we analyze the weighted cost for a portfolio of EPSs associated with separate hedging strategies, which were discussed in Section \ref{sec3}. We only consider here the case of one year maturity since the influence of maturity is likely to be similar as in the case of EPSs with a single currency, which was extensively studied by Xu et al. \cite{XLR2024}.

Table \ref{table2} presents the hedging cost and key parameters for several variants of domestic and foreign EPSs.
For brevity, we use the shorthand notation for various portfolios of domestic and foreign EPSs with all prices
expressed in the domestic currency for the notional principal equal to AUD 100: \\
(i)   Domestic -- the price $\EPSdo$ of the domestic EPS (Corollary \ref{cor3.2}); \\
(ii)  Nominal -- the weighted cost $w\,\EPSdo + (1-w)\X_0\EPSfo$ where $\EPSfo$ is the price of the foreign EPS
     (Corollary \ref{cor3.2a});\\
(iii) Effective -- the weighted cost $w\,\EPSdo + (1-w)\EPSfeo$ where $\EPSfeo$ is the price of the effective foreign EPS
      (Corollary \ref{cor3.4});\\
(iv)  Quanto foreign -- the weighted cost $w\,\EPSdo + (1-w)\EPSfqo$ where $\EPSfqo$ is the price of the quanto EPS (Corollary \ref{cor3.5}).

\begin{table}[h!]
\centering
\begin{tabular}{|c|c|cccc|c|c|c|c|c|}
\hline
Buffer &$w$&$l_1$&$g_1$&$p_2$&$f_2$& Domestic  & Nominal & Effective & Quanto  \\
\hline
[1] & 50\% & -5\% & 5\% & 0.5 & 0.5 & -0.114& -0.874& -0.503& -0.995 \\

[2] & 50\% & -5\% & 5\% & 0.8 & 0.5 & 0.421& 0.048& 0.197& -0.097 \\

[3] & 20\% & -5\% & 10\% & 0.5 & 0.5 & 0.426& -0.167& 0.023& -0.327 \\

[4] & 50\% & -5\% & 10\% & 0.5 & 0.5 & 0.426& 0.055& 0.174& -0.044 \\

[5] & 80\% & -5\% & 10\% & 0.5 & 0.5 & 0.426& 0.277& 0.325& 0.238 \\

[6] & 20\% & -5\% & 10\% & 0.8 & 0.5 & 0.961& 0.986& 0.822& 0.789\\

[7] & 50\% & -5\% & 10\% & 0.8 & 0.5 & 0.961& 0.976& 0.874& 0.853 \\

[8] & 80\% & -5\% & 10\% & 0.8 & 0.5 & 0.961& 0.967& 0.926& 0.918 \\

[9] & 50\% & -5\% & 10\% & 0.8 & 0.8 & 0.681& 0.088& 0.279& -0.071 \\

[10] & 20\% & -10\% & 10\% & 0.8 & 0.5 & 0.070& -0.504& -0.217& -0.663 \\

[11] & 50\% & -10\% & 10\% & 0.8 & 0.5 & 0.070& -0.289& -0.109& -0.388 \\

[12] & 80\% & -10\% & 10\% & 0.8 & 0.5 & 0.070& -0.073& -0.002& -0.113 \\

\hline
\hline

Floor &$w$&$l_1$&$g_1$&$p_1$&$f_2$& Domestic &Nominal &Effective &Quanto \\

\hline

[1] & 50\% & -5\% & 5\% & 0.8 & 0.5 & 0.591& -0.410& -0.088& -0.518 \\

[2] & 50\% & -5\% & 10\% & 0.5 & 0.5 & 0.532& -0.231& -0.004& -0.308 \\

[3] & 20\% & -5\% & 10\% & 0.8 & 0.5 & 1.131& 0.152& 0.264& 0.013 \\

[4] & 50\% & -5\% & 10\% & 0.8 & 0.5 & 1.131& 0.519& 0.589& 0.432 \\

[5] & 80\% & -5\% & 10\% & 0.8 & 0.5 & 1.131& 0.886& 0.914& 0.851 \\

[6] & 50\% & -5\% & 10\% & 0.8 & 0.8 & 0.851& -0.369& -0.006& -0.493 \\

[7] & 50\% & -10\% & 10\% & 0.8 & 0.5 & 2.022& 1.784& 1.573& 1.674 \\

[8] & 20\% & -15\% & 10\% & 0.5 & 0.5 & 1.330& 0.787& 0.700& 0.622 \\

[9] & 50\% & -15\% & 10\% & 0.5 & 0.5 & 1.330& 0.991& 0.936& 0.888 \\

[10] & 80\% & -15\% & 10\% & 0.5 & 0.5 & 1.330& 1.194& 1.172& 1.153 \\

[11] & 20\% & -15\% & 10\% & 0.8 & 0.5 & 2.407& 2.513& 1.905& 2.307 \\

[12] & 50\% & -15\% & 10\% & 0.8 & 0.5 & 2.407& 2.473& 2.093& 2.345 \\

[13] & 80\% & -15\% & 10\% & 0.8 & 0.5 & 2.407& 2.434& 2.282& 2.382 \\

[14] & 50\% & -15\% & 10\% & 0.8 & 0.8 & 2.127& 1.585& 1.498& 1.420 \\

\hline
\end{tabular}
\caption{Prices of EPSs on separate domestic and foreign returns}
\label{table2}
\end{table}

From Table \ref{table2}, the influence of protection and fee thresholds, as well as the associated participation rates,
are the same for separate hedging strategies as for single currency underlying asset, which has been discussed by Xu et al. \cite{XLR2024}.  With the setting of the predetermined exchange rate equal to the initial exchange rate, we can see that the price of the quanto EPS is usually close to the price of the corresponding foreign EPS with the same parameters, for both buffer and floor EPSs. We observe that the level of the volatility for the domestic and foreign assets has a profound effect on the price of the quanto EPS.

Moreover, the prices of effective EPS has the highest values among the three weighted hedging strategies. This was
expected since only the effective EPS involves the stochastic exchange rate and thus covers also the currency risk.
If the weight $w$ for the domestic asset increases, then the influence of hedging costs for foreign equities decreases
and thus the difference between costs of three variants of EPSs become smaller.

\begin{example} \label{ex6.1} {\rm
As an illustration of static hedging, we consider the buffer EPS No. 3 from Table \ref{table2} with the notional principal equal to one million AUD (i.e., AUD 1M) with 20\% invested in domestic assets and 80\% held in foreign assets. We assume that the domestic EPS and the effective foreign EPS mature in one year, have the buffer threshold $l_1=-5\%$ with the protection rate $p_2=50\%$ and the fee threshold $g_1=10\%$ with the fee rate $f_2=50\%$. It is clear that the initial prices of domestic and foreign assets are irrelevant in all pricing results (see Corollaries \ref{cor3.2}, \ref{cor3.2a}, \ref{cor3.4} and \ref{cor3.5}) but they are important when constructing the static hedging with European options. Here we assume that $\Sd_0$ and $\Sf_0$ correspond to market quotes of ASX 200 and S\&P 500, $\Sd_0=100$ AUD and $\Sf_0=70$ USD, and the exchange rate $\X_0=1.58$}.

{\rm For the domestic EPS, we find from Corollary \ref{cor3.2} that the fair premium equals $20\% \times 4260=852$, which is used to establish a static hedging strategy for the domestic component of the EPS using domestic call and put options written on $\Sd$, as explained in Corollary \ref{cor3.1}. For the foreign component, we use Corollary \ref{cor3.2a} to find that the provider can obtain AUD 216 by building the hedging portfolio using European options written on the process $\Sfd$ with the initial value $\Sfd_0=110.6$ AUD. Specifically, the static hedge is composed of:
\begin{itemize}
\item $\Nfd p_2/\Sfd_0$ long put options with strike $K^l_1(\Sfd_0):=(1+l_1)\Sfd_0=105.07$ AUD,
\item $\Nfd f_2/\Sfd_0$ short call options with strike $K^g_1(\Sfd_0):=(1+g_1)\Sfd_0=121.66$ AUD.
\end{itemize}
Hence to statically hedge the foreign component the provider needs to long $(0.8M) \times 0.5/110.6=3617$ units of put options with strike AUD 105.07 and short $(0.8M) \times 0.5/110.6=3617$ units of call options with strike AUD 121.66.
The fair premium for the combined position in two EPSs equals $20\% \times 4210 + 80\% \times (-777.5) =852-622=230$ AUD per one million of the notional principal.}
\end{example}

%%%%%%%%%%%%%%%%%%%%%%%%%%%%%%%%%%%%%%%%%%%%%%%%%%%%%%%%%%%%%%%%%%%%%%%%%%%%%%%%%%%%%%%%%%%%%
\subsection{Model Pricing of Basket Options}  \label{sec5.2.0}
%%%%%%%%%%%%%%%%%%%%%%%%%%%%%%%%%%%%%%%%%%%%%%%%%%%%%%%%%%%%%%%%%%%%%%%%%%%%%%%%%%%%%%%%%%%%%

Before applying the pricing formula of cross-currency EPSs, we analyzed the two approximation methods of basket options - geometric averaging (reported as GM) and moment matching (reported as MM). Table \ref{table_basket} reports basket option prices for different strikes 
$K=1.1,0.9, 1.0$, obtained using Monte Carlo simulation as the benchmark alongside the two approximation methods. The results show that both geometric averaging and moment matching methods deliver prices very close to the simulation values, with deviations typically below $0.01\%$ across all parameter settings. This indicates that both methods are accurate and reliable for basket option valuation.

\begin{table}[h!]
\centering
\begin{tabular}{|c|ccc|ccc|ccc|}
\hline
\multicolumn{1}{|c|}{Options} & \multicolumn{3}{|c|}{Parameters} & \multicolumn{3}{|c|}{Call($K=1.1$)} & \multicolumn{3}{|c|}{Put($K=0.9$)} \\
\hline
No. &$w$ &$\rho$ &T &Simulation & GM & MM &Simulation & GM & MM \\
\hline
[1]& 20\% & -0.4 & 1 & 0.01745 & 0.01661 & 0.01742 & 0.09695 & 0.09611 & 0.09691 \\

[2]& 20\% & 0.1 & 1 & 0.02041 & 0.01993 & 0.02040 & 0.09995 & 0.09943 & 0.09990 \\

[3]& 20\% & 0.7 & 1 & 0.02390 & 0.02373 & 0.02393 & 0.10350 & 0.10323 & 0.10343 \\

[4]& 50\% & -0.4 & 1 & 0.00468 & 0.00418 & 0.00468 & 0.09003 & 0.0895 & 0.0900 \\

[5]& 50\% & 0.1 & 1 & 0.00991 & 0.00958 & 0.00993 & 0.09529 & 0.09490 & 0.09525 \\

[6]& 50\% & 0.7 & 1 & 0.01603 & 0.01587 & 0.01609 & 0.10145 & 0.10118 & 0.10140  \\

[7]& 80\% & -0.4 & 1 & 0.00402 & 0.00385 & 0.00403 & 0.09523 & 0.09499 & 0.09517 \\

[8]& 80\% & 0.1 & 1 & 0.00707 & 0.00704 & 0.00709 & 0.09829 & 0.09818 & 0.09824 \\

[9]& 80\% & 0.7 & 1 & 0.01096 & 0.01093 & 0.01099 & 0.10219 & 0.10207 & 0.10214 \\

[10]& 20\% & -0.4 & 2 & 0.03855 & 0.03692 & 0.03852 & 0.09870 & 0.09706 & 0.09867 \\

[11]& 20\% & 0.1 & 2 & 0.04327 & 0.04226 & 0.04326 & 0.10347 & 0.10240 & 0.10341 \\

[12]& 20\% & 0.7 & 2 & 0.04865 & 0.04824 & 0.04871 & 0.10895 & 0.10839 & 0.10885 \\

[13]& 80\% & -0.4 & 2 & 0.01179 & 0.01146 & 0.01186 & 0.09463 & 0.09420 & 0.09461 \\

[14]& 80\% & 0.1 & 2 & 0.01774 & 0.01762 & 0.01780 & 0.10059 & 0.10037 & 0.10055 \\

[15]& 80\% & 0.7 & 2 & 0.02458 & 0.02445 & 0.02465 & 0.10745 & 0.10719 & 0.10740 \\

\hline
\multicolumn{1}{|c|}{Options} & \multicolumn{3}{|c|}{Parameters} & \multicolumn{3}{|c|}{Call($K=1$)} & \multicolumn{3}{|c|}{Put($K=1$)} \\
\hline
No. &$w$ &$\rho$ &T &Simulation & GM & MM &Simulation & GM & MM \\
\hline
[1]& 20\% & -0.4 & 1 & 0.05265 & 0.05239 & 0.05262 & 0.03616 & 0.03590 & 0.03613 \\

[2]& 20\% & 0.1 & 1 & 0.05650 & 0.05629 & 0.05649 & 0.04006 & 0.03980 & 0.04000 \\

[3]& 20\% & 0.7 & 1 & 0.06069 & 0.06061 & 0.06073 & 0.04431 & 0.04413 & 0.04424 \\

[4]& 50\% & -0.4 & 1 & 0.03344 & 0.03307 & 0.03347 & 0.02281 & 0.02240 & 0.02281 \\

[5]& 50\% & 0.1 & 1 & 0.04197 & 0.04171 & 0.04199 & 0.03137 & 0.03104 & 0.03133 \\

[6]& 50\% & 0.7 & 1 & 0.05008 & 0.04997 & 0.05013 & 0.03951 & 0.03931 & 0.03947  \\

[7]& 80\% & -0.4 & 1 & 0.03051 & 0.03054 & 0.03056 & 0.02573 & 0.02570 & 0.02572 \\

[8]& 80\% & 0.1 & 1 & 0.03627 & 0.03624 & 0.03631 & 0.03150 & 0.03140 & 0.03147 \\

[9]& 80\% & 0.7 & 1 & 0.04206 & 0.04205 & 0.04212 & 0.03731 & 0.03721 & 0.03728 \\

[10]& 20\% & -0.4 & 2 & 0.07752 & 0.07686 & 0.07749 & 0.04554 & 0.04487 & 0.04551 \\

[11]& 20\% & 0.1 & 2 & 0.08275 & 0.08220 & 0.08274 & 0.05082 & 0.05021 & 0.05076 \\

[12]& 20\% & 0.7 & 2 & 0.08838 & 0.08814 & 0.08845 & 0.05655 & 0.05615 & 0.05647 \\

[13]& 80\% & -0.4 & 2 & 0.04288 & 0.04290 & 0.04297 & 0.03359 & 0.03352 & 0.03359 \\

[14]& 80\% & 0.1 & 2 & 0.0508 & 0.05067 & 0.05086 & 0.04152 & 0.04129 & 0.04148 \\

[15]& 80\% & 0.7 & 2 & 0.05871 & 0.05860 & 0.05881 & 0.04946 & 0.04922 & 0.04943 \\

\hline
\end{tabular}
\caption{Prices of Basket options under alternative valuation methods}
\label{table_basket}
\end{table}

When comparing the two approaches, the moment matching method consistently produces prices that are slightly closer to the Monte Carlo benchmark than geometric averaging method, for both call and put options across all strikes. In contrast, geometric averaging method tends to slightly underestimate call prices and, in some cases, overestimate puts relative to the benchmark. Although the deviations are small, the systematic pattern indicates that MM more accurately captures the distributional characteristics of the basket and better reflects the hedging costs of these options.

The superior performance of moment matching method can be attributed to its use of higher-order moments to approximate the distribution of the basket, which allows it to more precisely replicate both the variance and skewness of the underlying assets. geometric averaging method, while computationally simpler, relies on a geometric mean approximation that may slightly distort the distribution, particularly when asset correlations and weights create asymmetry in the basket payoff. Despite this, both methods remain efficient and produce results very close to the benchmark simulation, making them suitable for practical applications.

Overall, although both geometric averaging and moment matching methods are accurate and computationally efficient, the moment matching method demonstrates a clear advantage in replicating Monte Carlo prices. This systematic edge makes moment matching method the preferred choice for basket option approximation in subsequent cross-currency EPS pricing, ensuring greater precision in both valuation and hedging analysis.

%%%%%%%%%%%%%%%%%%%%%%%%%%%%%%%%%%%%%%%%%%%%%%%%%%%%%%%%%%%%%%%%%%%%%%%%%%%%%%%%%%%%%%%%%%%%%
\subsection{Model Pricing of Cross-Currency EPSs}  \label{sec5.2}
%%%%%%%%%%%%%%%%%%%%%%%%%%%%%%%%%%%%%%%%%%%%%%%%%%%%%%%%%%%%%%%%%%%%%%%%%%%%%%%%%%%%%%%%%%%%%

Using once again the data reported in Table \ref{table1} we now analyze model prices of cross-currency and quanto EPSs, as given by Corollary \ref{cor4.1}. Recall that in both cases we deal with European basket options on domestic returns combined with either effective or quanto foreign returns. Since closed-form solutions for basket options with arithmetic averaging are not available in a lognormal model, we use instead four numerical approaches, which were studied in Section \ref{sec4}: \\
(i)   Simulation -- the Monte Carlo method applied to payoffs of various kinds of EPSs; \\
(ii)  Geometric  -- the closed-form solutions for basket options with geometric averaging, which were obtained in Propositions \ref{pro4.1} for effective;\\
(iii) Moments -- the approximate prices obtained by the moment matching method for the cross-currency EPS from Lemma \ref{lem4.3};\\
(iv) Super -- the superhedging costs $\Hwesb_0$ and $\Hwesf_0$ from Corollaries \ref{cor4.2} and \ref{cor4.3}.

\begin{table}[h!]
\centering
\begin{tabular}{|c|c|c|cc|c|c|c|c|c|}
\hline
Buffer &$w$&$\rho_{12}$&$l_1$&$g_1$ &Simulation & Geometric & Moments & Super  \\
\hline
[1] & 50\% & -0.4 & -5\% & 5\% & -0.045& -0.091& -0.112& 1.474 \\

[2] & 80\% & -0.4 & -5\% & 5\% & 0.111& 0.096& 0.081& 1.228 \\

[3] & 50\% & 0.7 & -5\% & 5\% & 0.249& 0.163& 0.159& 0.439  \\

[4] & 80\% & 0.7 & -5\% & 5\% &  0.341& 0.304& 0.304& 0.361 \\

[5] & 50\% & -0.4 & -5\% & 10\% & 0.416& 0.385& 0.369& 1.767 \\

[6] & 80\% & -0.4 & -5\% & 10\% & 0.533& 0.526& 0.511& 1.508 \\

[7] & 50\% & 0.1 & -5\% & 10\% & 0.674& 0.617& 0.610& 1.496 \\

[8] & 80\% & 0.1 & -5\% & 10\% & 0.738& 0.710& 0.712& 1.263  \\

[9] & 50\% & 0.7 & -5\% & 10\% & 0.910& 0.840& 0.836& 0.973  \\

[10] & 80\% & 0.7 & -5\% & 10\% & 0.923& 0.893& 0.893& 0.927  \\

[11] & 50\% & 0.1 & -10\% & 10\% & -0.098& -0.124& -0.140& 0.513 \\

[12] & 80\% & 0.1 & -10\% & 10\% & -0.013& -0.028& -0.029& 0.335  \\

[13] & 50\% & 0.7 & -10\% & 10\% & -0.051 & -0.096 & -0.105 & 0.210 \\

[14] & 80\% & 0.7 & -10\% & 10\% & 0.023 & 0.005 & 0.003 & 0.101  \\

\hline
\hline

Floor &$w$&$\rho_{12}$&$l_1$&$g_1$& Simulation & Geometric & Moments & Super \\

\hline

[1] & 50\% & -0.4 & -5\% & 5\% & 0.608& 0.543& 0.538& 2.893 \\

[2] & 50\% & 0.7 & -5\% & 5\% & 0.107& 0.040& 0.036& 0.679 \\

[3] & 50\% & -0.4 & -5\% & 10\% & 1.068& 1.019& 1.019& 3.187 \\

[4] & 80\% & -0.4 & -5\% & 10\% & 1.190& 1.167& 1.169& 3.001 \\

[5] & 50\% & 0.7 & -5\% & 10\% & 0.769& 0.717& 0.713& 1.513 \\

[6] & 80\% & 0.7 & -5\% & 10\% & 1.016& 0.995& 0.993& 1.668 \\

[7] & 50\% & -0.4 & -10\% & 10\% & 1.589& 1.509& 1.519& 3.328 \\

[8] & 50\% & 0.7 & -10\% & 10\% & 1.730& 1.653& 1.654& 2.156 \\

[9] & 20\% & -0.4 & -15\% & 10\% & 1.997& 1.857& 1.859& 3.596 \\

[10] & 50\% & -0.4 & -15\% & 10\% & 1.715& 1.628& 1.637& 3.348 \\

[11] & 80\% & -0.4 & -15\% & 10\% & 1.914& 1.880& 1.878& 3.099 \\

[12] & 20\% & 0.1 & -15\% & 10\% & 2.05& 1.908& 1.912& 3.265 \\

[13] & 50\% & 0.1 & -15\% & 10\% & 2.042& 1.948& 1.954& 3.917 \\

[14] & 80\% & 0.1 & -15\% & 10\% & 2.155& 2.115& 2.117& 3.198 \\

[15] & 50\% & 0.7 & -15\% & 10\% & 2.216& 2.124& 2.127& 2.437 \\

[16] & 80\% & 0.7 & -15\% & 10\% & 2.323& 2.282& 2.284& 2.359 \\

\hline
\end{tabular}
\caption{Prices of cross-currency EPSs under alternative valuation methods}
\label{table3}
\end{table}

\begin{table}[h!]
\centering
\begin{tabular}{|c|c|c|cc|c|c|c|c|c|}
\hline
Buffer &$w$&$\rho_{12}$&$l_1$&$g_1$ &Simulation & Geometric & Moments & Super  \\
\hline
[1] & 50\% & -0.4 & -5\% & 5\% & -0.182& -0.221& -0.247& 1.974 \\

[2] & 80\% & -0.4 & -5\% & 5\% & 0.045& 0.035& 0.016 & 1.442 \\

[3] & 50\% & 0.7 & -5\% & 5\% & 0.068& -0.016& -0.020 & 0.307 \\

[4] & 80\% & 0.7 & -5\% & 5\% & 0.267& 0.230& 0.231& 0.149 \\

[5] & 50\% & -0.4 & -5\% & 10\% & 0.295& 0.272& 0.251& 2.299 \\

[6] & 80\% & -0.4 & -5\% & 10\% & 0.465& 0.466& 0.445& 1.727 \\

[7] & 50\% & 0.1 & -5\% & 10\% & 0.531& 0.478& 0.470 & 1.979 \\

[8] & 80\% & 0.1 & -5\% & 10\% & 0.672& 0.646& 0.646 & 1.466 \\

[9] & 50\% & 0.7 & -5\% & 10\% & 0.752& 0.683& 0.679 & 1.333 \\

[10] & 80\% & 0.7 & -5\% & 10\% & 0.856& 0.825& 0.826 & 0.986 \\

[11] & 50\% & 0.1 & -10\% & 10\% & -0.174& -0.195& -0.213& 0.737 \\

[12] & 80\% & 0.1 & -10\% & 10\% & -0.042 & -0.055& -0.057 & 0.435 \\

[13] & 50\% & 0.7 & -10\% & 10\% & -0.153 & -0.195& -0.205 & 0.092 \\

[14] & 80\% & 0.7 & -10\% & 10\% & -0.014 & -0.032& -0.033 & 0.045 \\

\hline
\hline

Floor &$w$&$\rho_{12}$&$l_1$&$g_1$& Simulation & Geometric & Moments & Super \\

\hline

[1] & 50\% & -0.4 & -5\% & 5\% & 0.492& 0.427& 0.421& 3.716 \\

[2] & 50\% & 0.7 & -5\% & 5\% & 0.002& -0.064& -0.069& 0.743 \\

[3] & 50\% & -0.4 & -5\% & 10\% & 0.969& 0.919& 0.919& 4.041 \\

[4] & 80\% & -0.4 & -5\% & 10\% & 1.152& 1.129& 1.132& 3.363 \\

[5] & 50\% & 0.7 & -5\% & 10\% & 0.686& 0.636& 0.630& 1.769 \\

[6] & 80\% & 0.7 & -5\% & 10\% & 0.988& 0.967& 0.965 & 1.805 \\

[7] & 50\% & -0.4 & -10\% & 10\% & 1.409& 1.329& 1.340 & 4.242 \\

[8] & 50\% & 0.7 & -10\% & 10\% & 1.592& 1.515& 1.515 & 2.707 \\

[9] & 20\% & -0.4 & -15\% & 10\% & 1.715& 1.580& 1.579& 5.065 \\

[10] & 50\% & -0.4 & -15\% & 10\% & 1.502& 1.416& 1.425& 4.270 \\

[11] & 80\% & -0.4 & -15\% & 10\% & 1.807& 1.778& 1.771& 3.475 \\

[12] & 20\% & 0.1 & -15\% & 10\% & 1.775& 1.634& 1.636 & 4.636 \\

[13] & 50\% & 0.1 & -15\% & 10\% & 1.848& 1.754& 1.761& 3.917 \\

[14] & 80\% & 0.1 & -15\% & 10\% & 2.065& 2.024& 2.026& 3.198 \\

[15] & 50\% & 0.7 & -15\% & 10\% & 2.037& 1.943& 1.946 & 3.127 \\

[16] & 80\% & 0.7 & -15\% & 10\% & 2.242& 2.201& 2.203& 2.650 \\

\hline
\end{tabular}
\caption{Prices of quanto EPSs under alternative valuation methods}
\label{table4}
\end{table}

In Table \ref{table3} (resp., Table \ref{table4}) we report prices of cross-currency EPSs (resp., quanto EPSs) for AUD 100 notional principal computed using four alternative numerical approaches. The Monte Carlo estimates are based on one million simulations and thus the prices of cross-currency basket options, and hence also the prices of cross-currency EPSs, obtained through simulation are considered to be exact values. It is clear that the three numerical approaches for the hedging costs - simulation, geometric averaging, and three moment matching - provide very close results, with no more than 0.01\% difference in most cases. We also note that, when compared with the three moment matching method, the geometric average approximation provides prices closer to the exact prices obtained by the Monte Carlo method in most cases. Moreover, the hedging costs under three moment matching method are usually lower than the exact costs, for both the buffer and floor EPS.
%We also note that, when compared with the geometric average approximation, the three moments matching method provides prices closer to exact prices, which is consistent with findings from Hu et al. \cite{HSV2023} who observed that the moment matching technique yields more accurate results for basket options than the analytical technique based on geometric averaging.
As was expected, the superhedging prices are always higher than the exact prices since the terminal value of a superhedging strategy dominates the payoff of an EPS and the market model is arbitrage-free.

It should be noted that the prices of cross-currency and quanto EPSs with identical protection and fee parameters are very close. This is due to the fact that we use a quanto exchange rate $\overline{\X}=\X_0$ and the volatility of the exchange rate is relatively small (as compared to the volatilities of equities), so that the influence of the exchange rate is also relatively small. Of course, if a holder chooses a predetermined exchange rate markedly different from its current value
then the prices of cross-currency and quanto EPSs would be markedly different. It is worth noting that in Tables \ref{table3} and \ref{table4} the prices of cross-currency EPSs are higher than that of quanto EPSs when dealing with both buffer EPSs and floor EPSs, because the protection of a cross-currency EPS is more comprehensive than that of a quanto EPS with identical
parameters.

Comparing hedging and superhedging strategies, it is clear that a superhedging cost, which is reported as Super in Tables \ref{table3} and \ref{table4}, dominates the cost of a respective hedging and quanto hedging strategy in all cases.
This result is consistent with the property that the superhedging strategy covers more than the overall provider's
gains or losses and thus it should cost more than the respective hedging strategy.

Moreover, the structure of EPS products influences the difference between hedging and superhedging costs. We can see that the difference between hedging and superhedging costs for the floor EPS is higher than for the buffer EPS for both cross-currency
and quanto EPS. Since the floor EPS has a more complex structure than the buffer EPS it needs more long and short positions
in European options to hedge the respective final payoff and this increases the costs rapidly when applying superhedging strategies.

\begin{example} \label{ex6.2}
{\rm  As an example of a cross-currency EPS, we consider the floor EPS No. 6 from Table \ref{table3} with the notional principal one million AUD (i.e., AUD 1M), in which 80\% invested in domestic assets and 20\% held in foreign assets. The buffer threshold equals $l_1=-5\%$ with participation rate $p_2=80\%$ and the fee threshold equals $g_1=10\%$ with participation rate $f_2=50\%$. As usual, the maturity date is taken to be $T=1$. As shown in Section \ref{sec5}, assume the correlation between two assets is $\rhodf=0.7$, the arbitrage-free premium for the EPS equals AUD 10160.}

{\rm We first note that $\Swe_0=0.8 \times 100 + 0.2 \times 1.58 \times 70 =102.12$ AUD.
Hence a static hedging portfolio with European options written on $\Swe$ involves:
\begin{itemize}
\item $\frac{p_2}{\Swe_0}$ short put options with strike $K^l_1(\Swe_0):=(1+l_1)\Swe_0 = 97.014$,
\item $\frac{p_2}{\Swe_0}$ long put options with strike $\Swe_0= 102.12$,
\item $\frac{f_2}{\Swe_0}$ short call options with strike $K^g_1(\Swe_0):=(1+g_1)\Swe_0 =112.332$.
\end{itemize}
The provider needs to short $(0.8M/102.12)=7834$ units of put options with strike 97.014, long $(0.8M/102.12)=7834$ units of put options with strike 102.12, and short $(0.5M/102.12)=4896$ units of call options with strike 112.332.}
\end{example}

% \newpage

%%%%%%%%%%%%%%%%%%%%%%%%%%%%%%%%%%%%%%%%%%%%%%%%%%%%%%%%%%%%%%%%%%%%%%%%%%%%%%
%%%%%%%%%%%%%%%%%%%%%%%%%%%%%%%%%%%%%%%%%%%%%%%%%%%%%%%%%%%%%%%%%%%%%%%%%%%%%%
\subsection{Backtesting of Aggregated Cross-Currency EPS}       \label{sec6}
%%%%%%%%%%%%%%%%%%%%%%%%%%%%%%%%%%%%%%%%%%%%%%%%%%%%%%%%%%%%%%%%%%%%%%%%%%%%%%
%%%%%%%%%%%%%%%%%%%%%%%%%%%%%%%%%%%%%%%%%%%%%%%%%%%%%%%%%%%%%%%%%%%%%%%%%%%%%%

We now present a backtesting framework for the cross-currency EPS product. From UniSuper, we obtain daily cumulative historical returns for the Balanced, Growth, and High Growth portfolios from May 1, 2015, to April 30, 2025, along with daily returns for Australian shares, international shares, and cash. Additionally, we utilize daily closing prices of the S\&P~500 and ASX~200 indices to align with the daily returns of international and Australian shares, respectively.
We compute the correlations between the weekly returns of Australian shares and the ASX~200 index, and between international shares and the S\&P~500 index. Both correlations exceed 0.9, suggesting that these indices are appropriate proxies for simulating the historical returns of Australian and international shares in the UniSuper portfolios.

Our analysis assumes that a cohort of investors enters into one-year EPS contracts of a specific type on each trading day, with all contracts initiated on the same day having the same notional principal. Considering one-year EPS contracts initiated from May 1, 2015, to April 30, 2024, we generate 3288 one-year trailing returns for the Balanced, Growth, and High Growth portfolios. We apply the cross-currency basket EPS contracts to the High Growth portfolio to assess the effect of loss protection. The EPS-adjusted one-year trailing returns are then compared to the original realized returns of the Balanced, Growth, and High Growth portfolios. 

According to UniSuper \cite{Unisuper}, the High Growth portfolio consists of $43\%$ Australian shares, $49\%$ international shares, $5\%$ infrastructure and private equity, and $3\%$ property. For simplicity, we assume that $8\%$ of the portfolio is allocated to cash, with the remaining allocation invested in Australian and international shares. 

Unlike single-currency EPS products, the hedging strategy for cross-currency basket EPS contracts requires over-the-counter (OTC) basket options, for which market prices are not directly observable. Therefore, we employ a geometric averaging method to approximate the prices of these basket options, allowing us to estimate the fair fee rate for the cross-currency basket EPS.  
Under this dataset, we obtain the real interest rates in both Australia and the U.S. as the risk-neural interest rate for Australian and internationl shares. We also derive implied volatilities for the market indices using their respective VIX indices and incorporate historical exchange rate data for each trading date. However, as direct exchange rate volatility data is unavailable, we approximate it using the standard deviation of historical exchange rate changes, measured as $\frac{Q_T}{Q_0}$. Unlike index volatility, which varies over time, we assume exchange rate volatility remains constant. Additionally, we compute the correlations between weekly index returns and exchange rate returns over the holding period to model the dependence structure of the Wiener processes, in accordance with Eq.~\eqref{eq_Wiener}.
The parameters used in this analysis are presented in the following table:

\begin{table}[h]
\centering
\begin{tabular}{|l|c|}
\hline
Parameters & Value \\
\hline
Domestic short-term interest rate $r^d$ & Australian interest rate \\
Foreign short-term interest rate $r^f$ & U.S. interest rate \\
Weights in domestic equity $w_1$ & 0.43 \\
Weights in foreign equity $w_2$ & 0.49 \\
Weights in cash $1-w_1-w2$ & 0.08 \\
Volatility of the domestic equity $\|\bm{\sigma}^d \|$ & ASX VIX \\
Volatility of the foreign equity $\|\bm{\sigma}^f \|$ & SP VIX \\
Volatility of the exchange rate $\|\bm{\sigma}^q \|$ & 6.74\% \\
\hline
\end{tabular}
\caption{Parameters of backtesting}
\label{tab:back_paras}
\end{table}

In our empirical study of the real-world benefits of an EPS for its holder, we will focus on the following six specifications for the buffer and floor EPSs:
\begin{itemize}
    \item \textbf{Buffer 1}: $p_2=0.5,\ l_1=-5\%,\ g_1=5\%$;
    \item \textbf{Buffer 2}: $p_2=0.7,\ l_1=-5\%,\ g_1=5\%$;
    \item \textbf{Buffer 3}: $p_2=0.7,\ l_1=-10\%,\ g_1=10\%$;
    \item \textbf{Floor 1}:  $p_1=0.5,\ l_1=-10\%,\ g_1=5\%$;
    \item \textbf{Floor 2}:  $p_1=0.7,\ l_1=-10\%,\ g_1=5\%$;
    \item \textbf{Floor 3}:  $p_1=0.5,\ l_1=-10\%,\ g_1=10\%$.
\end{itemize}
For each cross-currency basket EPS, the fair fee rate is calculated consistently with the current market data (using the geometric averaging method for basket options) so it depends on their inception date and the structure of an EPS. We first make a selection of the protection rate ($p_1$ or $p_2$) and the protection and fee thresholds ($l_1$ and $g_1$), and then we use the market data for basket option approximation (according to Proposition \ref{pro4.1}) to identify a unique level of the fair fee rate $f_2$.

We have the following figures to show the empirical densities of the protected trailing returns of the High Growth portfolio with cross-currency basket EPS. The original realised one-year trailing returns of the High Growth portfolio are also summarised by empirical density in the figures as benchmarks.
We also provide the mean value of relative gain, loss and net outcome (comparing with original returns) from each EPS product, and the Sharpe ratio of the original and protected trailing returns. All values reported in tables, except for the Sharpe ratio, are given as percentages.

\begin{figure}[h]
    \centering
    \includegraphics[width=0.9\linewidth]{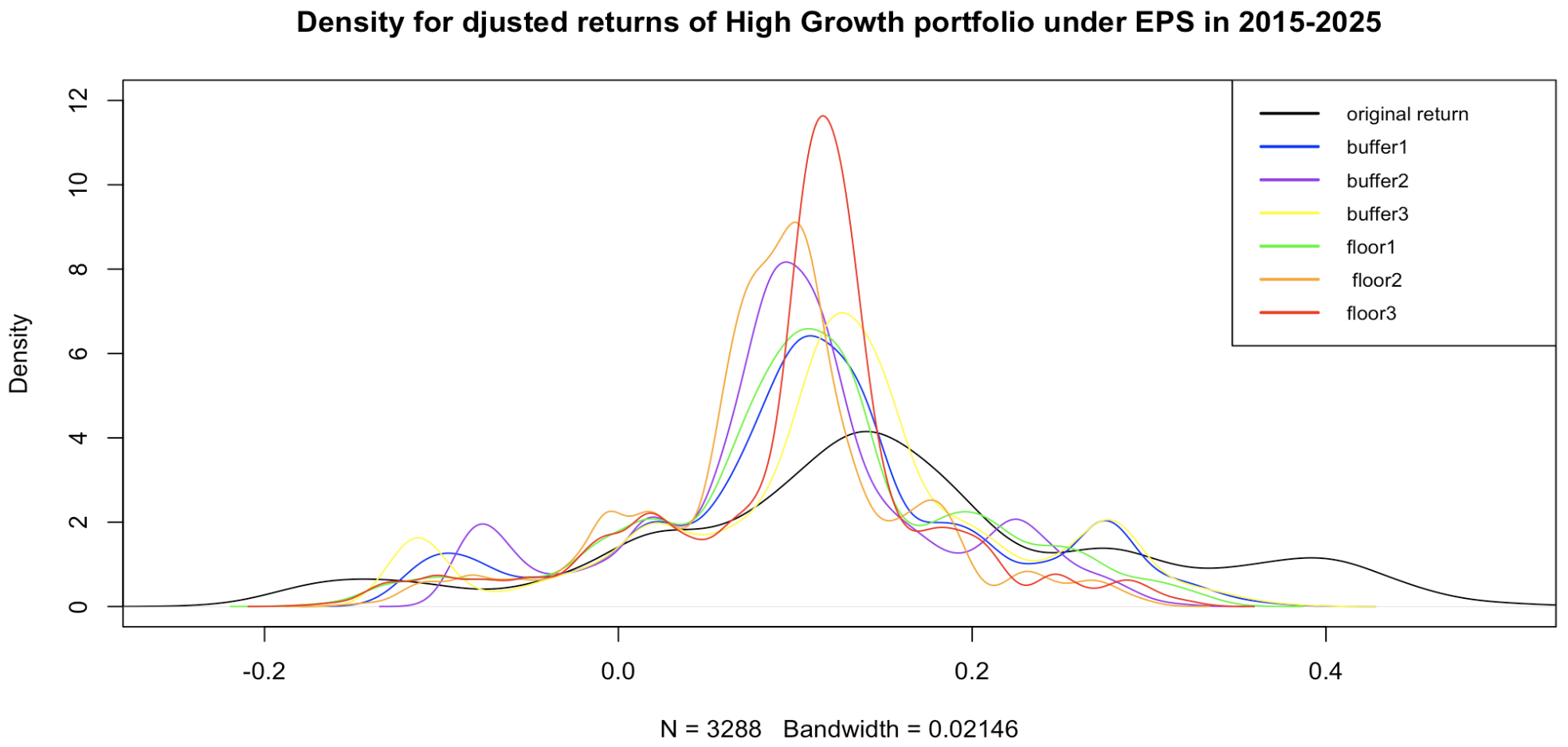}
    \caption{Empirical densities of returns on High Growth portfolio in 2015–2025}
    \label{fig:CC_backtest}
\end{figure}

\begin{table}[h]
\centering
\begin{tabular}{@{}c|ccccccc|cccc@{}}
\hline
\multicolumn{1}{c|} {CC} & \multicolumn{7}{c|}{Empirical quantiles} & \multicolumn{4}{c}{Empirical means and SR} \\
\hline
Case      & Min    & 10\%   & 25\%   & 50\%   & 75\%   & 90\%   & Max  & Gain  & Loss  & Net & Sharpe   \\
\hline
\hline

Original & -23.35 &-2.26 &6.94 &14.33 &23.08 &36.30 &55.41 &- &- &- & 0.9495 \\
\hline
Buffer1 & -14.17 &-2.26 &6.35 &11.08 &15.81 &26.38 &39.03 &4.12 &-5.32 &-3.85 & 0.6808 \\
\hline
Buffer2 & -10.50 &-2.26 &6.02 &9.72 &13.51 &21.71 &32.47 &3.44 &-5.03 &-3.16 & 0.5732 \\
\hline
Buffer3 & -14.00 &-2.26 &6.94 &12.55 &16.88 &26.64 &38.86 &5.76 &-7.45 &-5.38 & 0.7289 \\
\hline
Floor1 & -18.35 &-1.13 &6.20 &10.70 &15.18 &22.67 &35.09 &3.51 &-8.21 & -5.08 & 0.6402 \\
\hline
Floor2 & -16.35 &-0.68 &5.91 &9.20 &12.15 &17.80 &30.92 &4.91 &-8.65 &-6.20 & 0.5164 \\
\hline
Floor3 & -18.35 &-1.13 &6.92 &11.24 &13.38 &18.90 &33.36 &3.51 &-6.18 & -4.43 & 0.5947 \\
\hline
\end{tabular}
\caption{Performance of EPSs on High Growth portfolio in 2015-2025}
\label{quan_CC_backtest}
\end{table}

Figure \ref{fig:CC_backtest} and Table \ref{quan_CC_backtest} analyze the performance of High Growth portfolio with different cross-currency basket EPS products from 2015 to 2025, using empirical quantiles, empirical gains/losses, and Sharpe ratios. The density plot in Figure \ref{fig:CC_backtest} illustrates the empirical return distributions for the original portfolio and with various EPS products. The original return distribution appears more spread out, while the buffer and floor strategies alter the shape, reducing extreme fluctuations. Table \ref{quan_CC_backtest} provides empirical quantiles and performance metrics, it shows that the original portfolio exhibits the highest Sharpe ratio (0.9495), indicating superior risk-adjusted performance, but it also experiences the most extreme minimum (-23.35\%) and maximum (55.41\%) returns. Buffer strategies, especially Buffer2 and Buffer3, significantly limit extreme losses (e.g., Buffer2 minimum return is -10.50\%) while maintaining reasonable upside, though at the cost of lower Sharpe ratios. Among the floor strategies, Floor3 offers a balanced performance with improved lower quantiles and reduced losses relative to other floor EPSs. Overall, EPS strategies trade off some upside for better downside protection and smoother risk-adjusted returns, with Buffer3 achieving a relatively favorable balance between protection and performance.

While all EPS structures effectively reduce extreme downside risk—evident in improved minimum returns compared to the original portfolio—the Sharpe ratios and empirical gains/losses alone do not offer compelling justification for adopting these products. In fact, EPS strategies generally exhibit lower Sharpe ratios than the unprotected portfolio, and the net gain/loss values appear mixed across different configurations. Therefore, to better capture the value of downside protection and upside participation over time, we provide the time-series plots in Figure~\ref{fig:eps_comparison}, which more clearly illustrate the performance advantages of EPS-protected portfolios in both negative and positive market conditions.

Figure \ref{fig:eps_comparison} illustrates the performance of the High Growth portfolio with and without EPS protection using Buffer and Floor strategies, in comparison to the Balanced and Growth portfolios. During periods of negative returns, both EPS protected strategies —Buffer and Floor EPS — demonstrate clear downside protection, consistently outperforming the unprotected High Growth and Growth portfolios. This highlights the loss-mitigation benefits of the EPS structure. Conversely, in positive return environments, the EPS protected High Growth portfolios still outperform the Balanced portfolio, indicating that these strategies retain a meaningful degree of upside participation. Overall, the results suggest that applying EPS contracts to the High Growth portfolio can improve return outcomes in both adverse and favorable market conditions.

\begin{figure}[h!]
    \centering
    % First subfigure
    \begin{subfigure}[b]{0.49\textwidth}
        \centering
        \includegraphics[width=\textwidth]{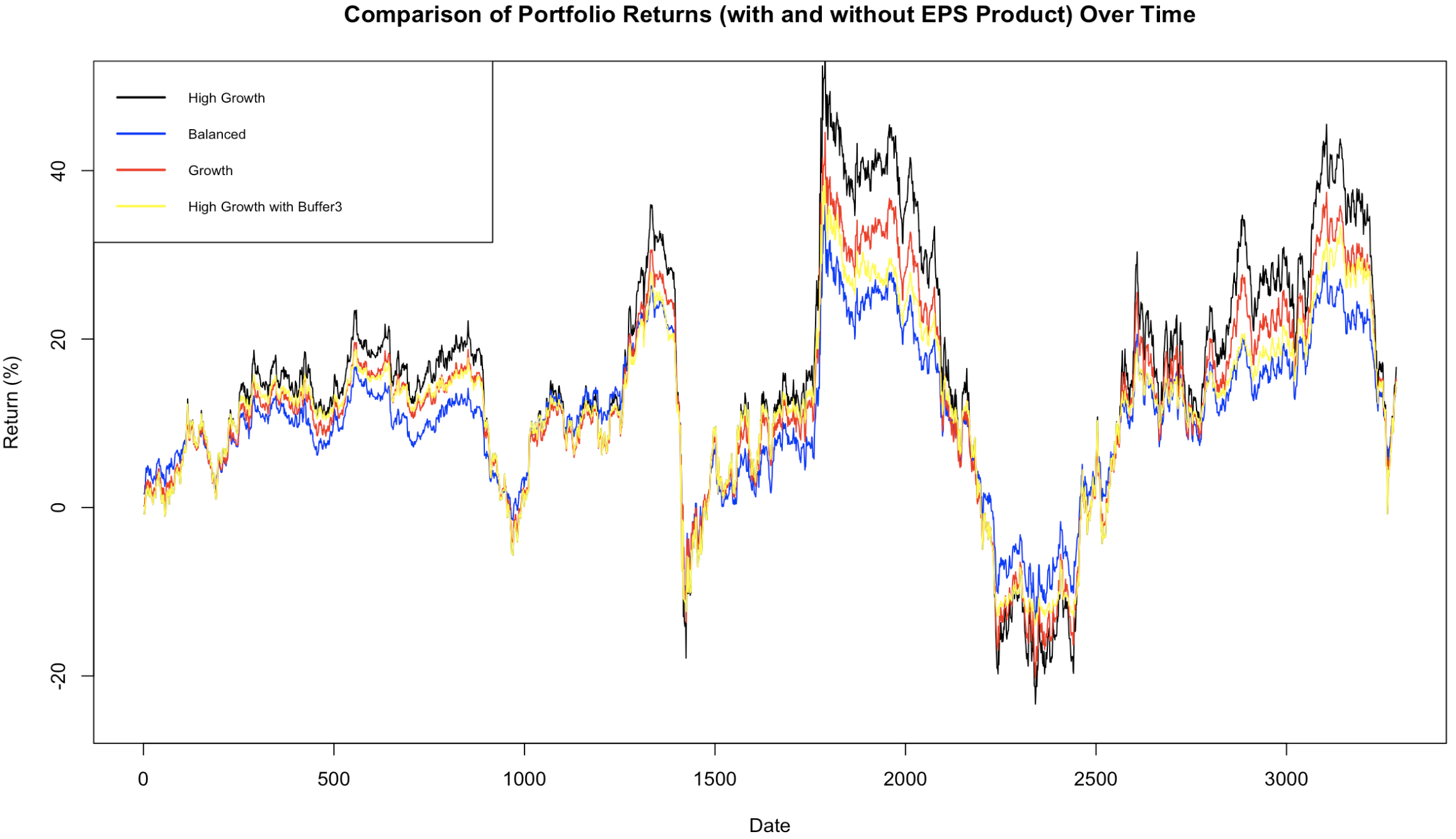}
        \caption{Buffer EPS}
        \label{fig:buffer3}
    \end{subfigure}
    \hfill
    % Second subfigure
    \begin{subfigure}[b]{0.49\textwidth}
        \centering
        \includegraphics[width=\textwidth]{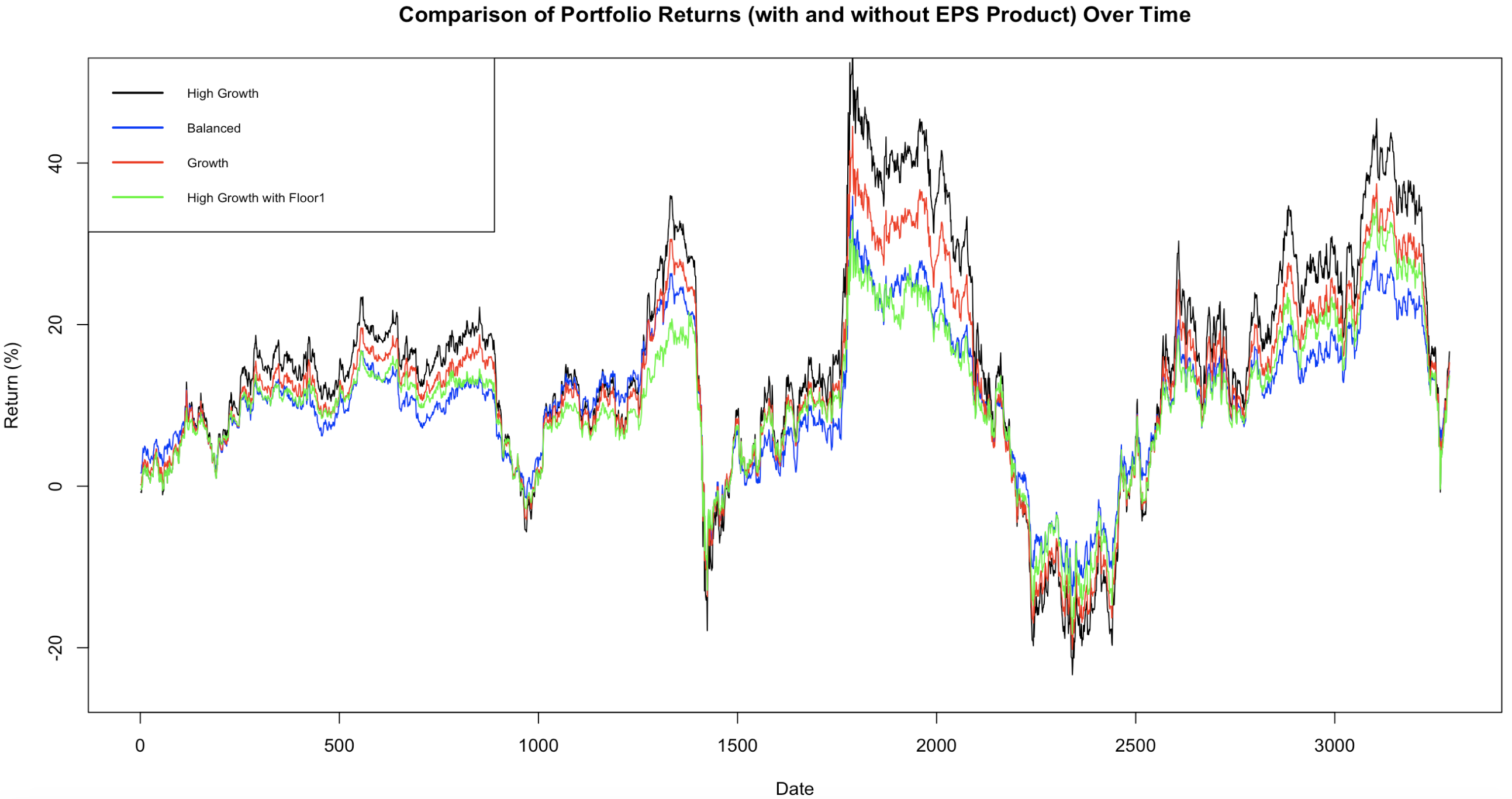}
        \caption{Floor EPS}
        \label{fig:floor1}
    \end{subfigure}
    \caption{Comparison of portfolio returns with and without cross-currency basket EPS}
    \label{fig:eps_comparison}
\end{figure}

Moreover, we provide another figure to present the comparison of the Omega ratios for the adjusted returns of the High Growth portfolio with the cross-currency basket EPS, alongside the returns of the Balanced and Growth portfolios. The Omega ratio is a performance measure that captures the likelihood of achieving returns above a specified threshold relative to the likelihood of falling below it, offering a comprehensive view of risk and reward beyond traditional metrics. According to \cite{Keating2002}, the evaluation statistic Omega has a precise mathematical definition as:
\begin{equation*}
\Omega(r) = \frac{\int_r^b (1-F(x))\,dx}{\int_a^r F(x)\,dx},    
\end{equation*}
where $(a,b)$ is the interval of returns and $F$ is the cumulative distribution of returns.

\begin{figure}[h]
    \centering
    \includegraphics[width=0.9\linewidth]{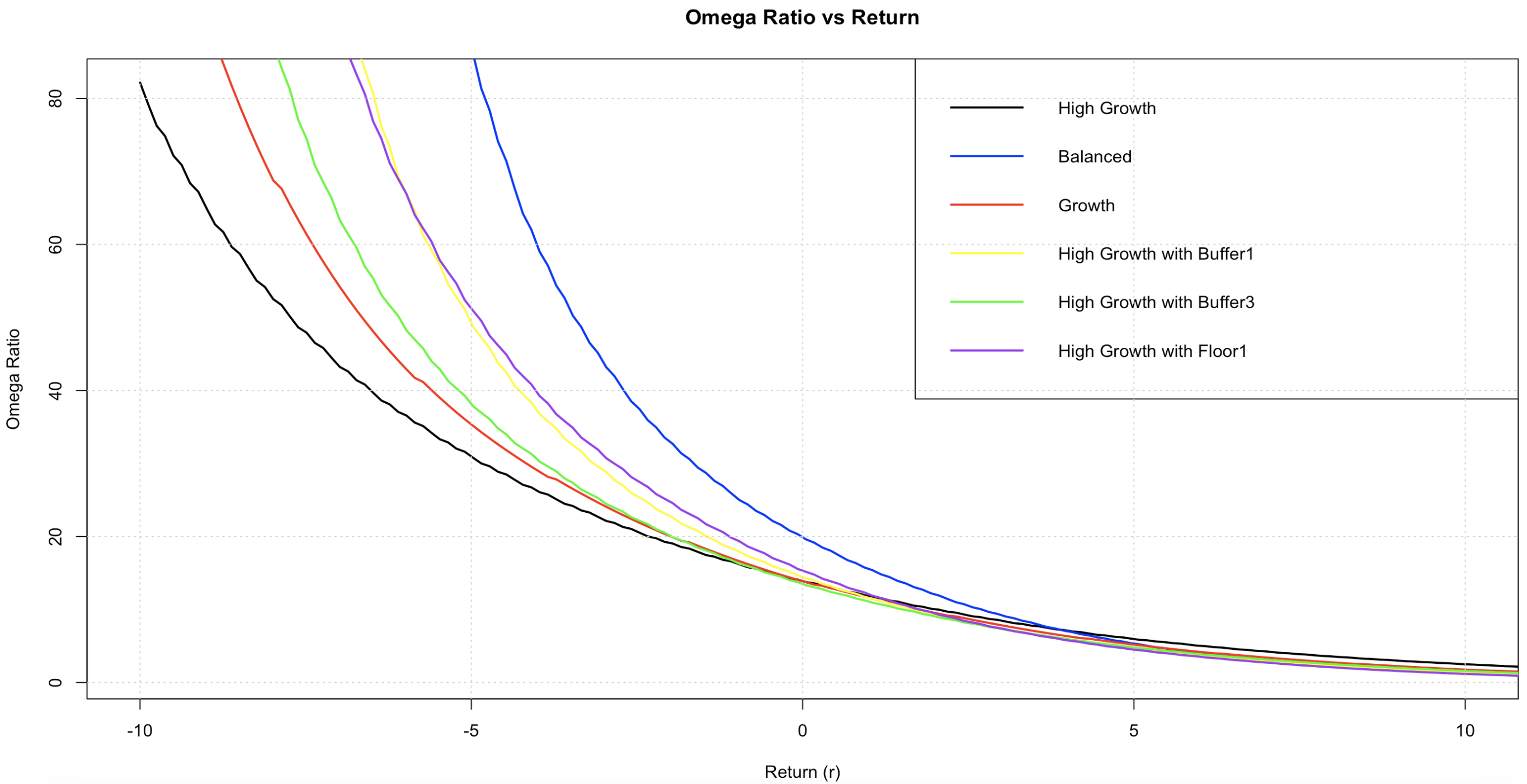}
    \caption{Comparison of Omega ratios for returns with and without cross-currency basket EPS}
    \label{fig:omega}
\end{figure}

Figure \ref{fig:omega} evaluates the performance of different investment strategies using the Omega ratio, a metric where higher values indicate better risk-adjusted returns. The High Growth portfolio with cross-currency basket EPS consistently achieves higher Omega ratios compared to the realized returns of the Growth and High Growth portfolios when returns are negative. This demonstrates its effectiveness in limiting downside risk while maintaining competitive performance.
Although the Balanced portfolio exhibits higher Omega ratios under negative return thresholds, its performance deteriorates as returns increase — the Omega ratio declines sharply, resulting in all portfolios showing similar ratios in the positive return region. In fact, at higher return levels, the Balanced portfolio's Omega ratio can fall below that of the High Growth portfolio with EPS.
The presence of buffers and floors in the structured EPS strategies helps mitigate downside risk, making them attractive to investors seeking growth with controlled volatility. These embedded protections enhance performance stability and offer a more favorable balance between risk and return than traditional, unprotected high-growth strategies.

%%%%%%%%%%%%%%%%%%%%%%%%%%%%%%%%%%%%%%%%%%%%%%%%%%%%%%%%%%%%%%%%%%%%%%%%%%%%%%%
%%%%%%%%%%%%%%%%%%%%%%%%%%%%%%%%%%%%%%%%%%%%%%%%%%%%%%%%%%%%%%%%%%%%%%%%%%%%%%%
\section{Conclusions}             \label{nsec6}
%%%%%%%%%%%%%%%%%%%%%%%%%%%%%%%%%%%%%%%%%%%%%%%%%%%%%%%%%%%%%%%%%%%%%%%%%%%%%%%
%%%%%%%%%%%%%%%%%%%%%%%%%%%%%%%%%%%%%%%%%%%%%%%%%%%%%%%%%%%%%%%%%%%%%%%%%%%%%%%

In this paper, the studies of {\it equity protection swaps} (EPS) from Xu et al. \cite{XLR2024} have been expanded into a more general situation - with consideration of cross-currency reference portfolio. An EPS is an insurance product for the variable annuity that provides protection against potential losses and collects insurance fees from actual gains. Compared with guarantees attached to variable annuities, which are offered in the U.S. market, an EPS has a simple structure and, more importantly, is a standalone financial derivative, which can be purchased independently of the variable annuity itself. We contend that it is a suitable product for the superannuation market in Australia where a typical holder of a pension account is able to directly invest in domestic and foreign securities, with 3.5 trillion Australian dollars (AUD) of total assets currently held in superannuation accounts. It should be stressed that an EPS can be structured in such a way that its initial cost is null and a fee is only charged when a reference portfolio generates significant profits during the lifetime of an EPS.

According to Xu et al. \cite{XLR2024}, the EPS products in a single economy can be perfectly hedged by a static portfolio of European call and put options. We have shown here that it is still possible to use European options to hedge the cross-currency EPS product but the form of a static hedging should be adapted to an EPS at hand and the use of basket options may be required for some cross-currency products. We first examine an EPS with separated protections for domestic and foreign portfolios, including nominal, effective, and quanto foreign returns. Then the basic algorithm is similar to the standard EPS product in a single currency, as presented in Xu et al. \cite{XLR2024}. A comparison of hedging costs for EPSs with three types of foreign equity returns shows that the hedging cost for foreign equity's effective return is the highest one, which was expected since an EPS of that kind provides the most comprehensive protection against currency and equity risk. Moreover, when a quanto exchange rate is chosen to be equal to the initial exchange rate, one expects that the hedging cost for nominal foreign return and quanto return should be close to each other and this is indeed the case, as can be seen from numerical results. Therefore, if a buyer of an EPS has strong views about the future level of the exchange rate, then the EPS referencing the quanto return could be a suitable product.

Another important class of cross-currency products are EPSs referencing the aggregated return on bespoke cross-currency portfolios.
We examine the case of the cross-currency effective return, as well as the quanto return, and we use cross-currency basket options to hedge EPS products on aggregated returns. To this end, we present three different methods to pricing the basket option that is needed for EPS product valuation, namely, the geometric averaging method, moments matching technique, and Monte-Carlo simulation. We provide a numerical study to compare the approximate basket option prices with simulation results, which are considered to be exact values. It can be seen that the hedging costs for cross-currency EPS products are fairly close under three valuation methods but we can see that the moment matching technique improves on the geometric averaging method and provides more accurate basket option prices.
Since cross-currency basket options on bespoke portfolios are not likely to be easily available in practice, we also examine superhedging strategies based solely on the same domestic and foreign market options as in the case of separated protections.

Finally, we provide numerical studies for model pricing and the backtesting of the performance of aggregated cross-currency EPSs. 
Three methods of pricing cross-currency EPSs are analyzed in the numerical study, including Monte Carlo simulation, geometric averaging, and moment matching, yield highly consistent estimates for cross-currency and quanto EPSs, with deviations from the Monte Carlo benchmark generally below 0.01\%.
From backtesting, we can see that cross-currency EPSs enhance portfolios by reducing downside risk while maintaining meaningful upside participation. Although Sharpe ratios remain lower than those of unprotected portfolios, the Omega ratio demonstrates superior risk-adjusted performance under negative return thresholds. These results suggest that EPS products provide a more resilient risk–return profile, making them valuable for growth-oriented investors seeking downside protection without fully sacrificing return potential.

\end{document}